\documentclass[letterpaper,11pt]{article}

\usepackage{amsmath}
\usepackage{amssymb}
\usepackage{amsthm}
\usepackage{mathtools}
\usepackage{xcolor}
\usepackage{complexity}
\usepackage{fullpage}
\usepackage[whole]{bxcjkjatype}
\usepackage{microtype}
\usepackage{enumitem}
\usepackage{thm-restate}
\usepackage{epigraph}
\usepackage{authblk}
\usepackage{hyperref}
\usepackage{cleveref}
\usepackage[
  backend=biber,
  style=alphabetic,
  natbib=false,
  maxnames=99,
  maxalphanames=4,
  minalphanames=3,
  useprefix=true
]{biblatex}

\hypersetup{
  linkcolor={red!75!black},
  citecolor={blue!75!black},
  urlcolor={red!75!black},
}

\theoremstyle{plain}
\newtheorem{theorem}{Theorem}[section]
\newtheorem{lemma}[theorem]{Lemma}
\newtheorem{proposition}[theorem]{Proposition}
\newtheorem{corollary}[theorem]{Corollary}

\theoremstyle{definition}
\newtheorem{definition}[theorem]{Definition}

\theoremstyle{remark}

\newcommand{\problem}[1]{\ensuremath{\mathrm{#1}}}
\renewcommand{\vec}[1]{\ensuremath{\boldsymbol{#1}}}
\DeclarePairedDelimiter\abs{|}{|}
\DeclarePairedDelimiter\norm{\lVert}{\rVert}
\newcommand{\dist}{\mathrm{dist}}
\newcommand{\SVP}{\problem{SVP}}
\newcommand{\GapSVP}{\problem{GapSVP}}
\newcommand{\coGapSVP}{\problem{coGapSVP}}
\newcommand{\GapCVP}{\problem{GapCVP}}
\newcommand{\LDLP}{\problem{LDLP}}
\newcommand{\CosetLDLP}{\problem{CosetLDLP}}
\newcommand{\CoeffLDLP}{\problem{CoeffLDLP}}

\usepackage[draft,multiuser,inline,nomargin]{fixme}
\usepackage[left=1in, right=1in, top=1in]{geometry}

\begin{document}

\title{On the Complexity of Locally Dense Lattices}

\author[1]{Shuichi Hirahara}
\author[2]{Kazuki Ogitsuka}
\affil[1]{National Institute of Informatics, Tokyo, Japan \protect\\ \texttt{\href{mailto:s_hirahara@nii.ac.jp}{s\_hirahara@nii.ac.jp}}}
\affil[2]{The Graduate University for Advanced Studies, SOKENDAI, Tokyo, Japan \protect\\ \texttt{\href{mailto:ogitsuka@nii.ac.jp}{ogitsuka@nii.ac.jp}}}

\date{}
\maketitle

\begin{abstract}
\emph{Locally dense lattices} are central gadgets used to prove the hardness of the Shortest Vector Problem and related lattice problems.
Informally, a locally dense lattice is a lattice $\mathcal{L}$ that contains exponentially many lattice vectors inside some $\ell_p$ ball centered at $\vec{s}$ with radius at most an $\alpha < 1$ fraction of the length of its shortest nonzero lattice vector.

In this paper, taking a ``meta'' viewpoint on locally dense lattices, we introduce the \emph{Locally Dense Lattice Problem} (LDLP), the decision problem of determining whether a given input specifies a locally dense lattice.
Our main result is that LDLP in $\ell_p$ norms for all finite $p \geq \log_2 3$ and for the infinity norm is complete for the second level of the polynomial hierarchy.

We also compare two standard definitions of local density that appear in prior work. 
Micciancio's original definition (FOCS 1998 and SICOMP 2001) uses integer coefficient vectors, while later work by Micciancio (ToC 2012) and by Bennett and Peikert (RANDOM 2023) uses short vectors in a shifted coset. 
We show that the corresponding promise problems are mutually reducible in deterministic polynomial time, which shows that the two formulations are robust.
\end{abstract}

\newpage
\section{Introduction}
A \emph{lattice} is the set of all integer linear combinations of some $n$ linearly independent vectors $\vec{b}_1,\ldots,\vec{b}_n \in \mathbb{R}^m$, i.e.,
\begin{align*}
    \mathcal{L} = \mathcal{L}(\vec{b}_1,\ldots,\vec{b}_n) \coloneq \left\{\sum^n_{i=1} a_i\vec{b}_i\ :\ a_1,\ldots,a_n\in \mathbb{Z}\right\}.
\end{align*}
The set of vectors $\vec{b}_1,\ldots,\vec{b}_n$ is called a \emph{basis} of $\mathcal{L}$. 
A basis can be represented by the matrix $\mathbf{B} := (\vec{b}_1,\ldots,\vec{b}_n)$, and the lattice generated by $\mathbf{B}$ is denoted $\mathcal{L}(\mathbf{B})$.
Lattices have been studied in mathematics and theoretical computer science,
and have found applications in cryptography.
Lattice-based cryptography, one of the most promising candidates for post-quantum cryptography, is based on the conjectured hardness of computational problems that concern lattices.
We refer the reader to the survey by Peikert \cite{10.1561/0400000074/Pei2016} on lattice-based cryptography.

\paragraph*{Locally Dense Lattices.}
One of the most fundamental computational problems on lattices is the Shortest Vector Problem (SVP). The task of SVP is to decide whether $\mathcal{L}(\mathbf{B})$ contains some non-zero vector of length at most $d$ or not for a given basis $\mathbf{B}$ and a given parameter $d$.%
\footnote{SVP is sometimes used for the associated search problem; throughout this paper, we use SVP for the decision version.}
The length of a vector is usually measured by the $\ell_2$ norm.
More generally, the version of SVP in which the length is measured by the $\ell_p$ norm is denoted by $\SVP_p$ for $p \in [1, \infty]$.

A line of hardness reductions for $\SVP_p$, initiated by Micciancio \cite{doi:10.1137/S0097539700373039/Mic01,journals/toc/Micciancio12} and further developed in subsequent work \cite{bennett_et_al:LIPIcs.APPROX/RANDOM.2023.37,wan2026}, uses a gadget called a \emph{locally dense lattice}.
Informally, such a gadget consists of a lattice $\mathcal{L}$, a shift vector $\vec{s}$, a radius bound \(r \le \alpha \cdot \lambda_1^{(p)}(\mathcal{L})\), and a linear map $\mathbf{T}$ such that, for every Boolean vector \(\vec{u}\in\{0,1\}^k\), there exists a vector \(\vec{x}\in(\vec{s}+\mathcal{L})\cap\mathcal{B}_p(r)\) satisfying \(\mathbf{T}\vec{x}=\vec{u}\).
Here $\alpha < 1$ is a constant, and $\lambda_1^{(p)}(\mathcal{L})$ denotes the minimum $\ell_p$ norm of non-zero vectors in $\mathcal{L}$.

There are two closely related formulations of locally dense lattices in the literature.
Both use a linear map \(\mathbf{T}\) to certify the presence of many short vectors, but they differ in what is used as the preimage under \(\mathbf{T}\).

Micciancio's original reduction \cite{doi:10.1137/S0097539700373039/Mic01} uses a \emph{coefficient} formulation, which asks for
\[
\{0,1\}^k \subseteq
\{\mathbf{T} \vec{x} : \vec{x} \in \mathbb{Z}^n,\ 
\|\mathbf{A}\vec{x}-\vec{s}\|_p \le \alpha \ell\}.
\]
We refer to this as a \((p,\alpha,k)\)-\emph{coefficient-locally dense lattice}.
Later work of Micciancio \cite{journals/toc/Micciancio12}, Bennett and Peikert \cite{bennett_et_al:LIPIcs.APPROX/RANDOM.2023.37}, and Wan \cite{wan2026} is instead phrased in terms of short vectors in a shifted coset, together with a linear map on those coset vectors:
\[
\{0,1\}^k \subseteq
\{\mathbf{T} \vec{v} : \vec{v} \in
(\vec{s}+\mathcal{L}(\mathbf{A})) \cap \mathcal{B}_p(\alpha \ell)\}.
\]
We call this a \((p,\alpha,k)\)-\emph{coset-locally dense lattice}.

In both formulations, the inclusion of \(\{0,1\}^k\) implies that the small ball contains at least \(2^k\) distinct vectors.
Indeed, preimages of two distinct Boolean vectors must be distinct because \(\mathbf{T}\) is a function.
In the coefficient formulation, the column rank of \(\mathbf{A}\) implies that distinct coefficient vectors yield distinct lattice points; in the coset formulation, the preimages are the coset vectors themselves.
For coefficient-locally dense lattices, Micciancio \cite{doi:10.1137/S0097539700373039/Mic01} constructed locally dense lattices using Schnorr--Adleman prime lattices.
For coset-locally dense lattices, Micciancio \cite{journals/toc/Micciancio12} constructed locally dense lattices via ``Construction D'' \cite{Conway99_book} applied to BCH codes, while Bennett and Peikert \cite{bennett_et_al:LIPIcs.APPROX/RANDOM.2023.37} did so via ``Construction A'' \cite{Conway99_book} applied to Reed--Solomon codes.
Randomized constructions are known for both types of locally dense lattices.
Recently, Wan \cite{wan2026} gave a deterministic construction of coset-locally dense lattices, derandomizing the construction of Bennett and Peikert based on ``Construction A'' applied to Reed--Solomon codes.
Using this deterministic construction, Wan proved that \(\SVP_p\) is \(\NP\)-hard under deterministic polynomial-time reductions.
This raises the question of whether deterministic constructions also exist for coefficient-locally dense lattices, and whether the two definitions of locally dense lattices are robust.

\subsection{Our results}
In this paper, we initiate the study of locally dense lattices from a ``meta'' viewpoint.
In the language of explicit constructions \cite{Korten21_focs_conf}, prior work \cite{doi:10.1137/S0097539700373039/Mic01,journals/toc/Micciancio12, bennett_et_al:LIPIcs.APPROX/RANDOM.2023.37,wan2026} has largely studied locally dense lattices from the perspective of constructing them explicitly.
Taking a step back, we introduce the following \emph{decision} problem, which asks to decide whether an input is a locally dense lattice or not (see \Cref{def:ldlp} for the formal definition).%
\footnote{This work was initiated before the breakthrough result by Wan \cite{wan2026} that proves the deterministic $\NP$-hardness of SVP. Our original motivation for introducing LDLP was an approach for proving the deterministic $\NP$-hardness of SVP; see \Cref{appendix:meta}.}
\begin{quote}
    \textbf{$(p,\alpha)$-Locally Dense Lattice Problem ($(p,\alpha)$-LDLP)}
    
    Given as input a tuple $(\mathbf{A}, \vec{s}, \ell, \mathbf{T})$,
    does it hold that $\lambda^{(p)}_1(\mathcal{L}(\mathbf{A})) \geq \ell$ and $\{0,1\}^k\subseteq\{\mathbf{T}\vec{v} : \vec{v}\in (\vec{s} + \mathcal{L}(\mathbf{A})) \cap \mathcal{B}_p(\alpha\cdot\ell)\}$?
\end{quote}

As shown in \Cref{prop:pi2p}, $(p,\alpha)$-$\LDLP$ belongs to $\Pi_2^p=\coNP^{\NP}$ for every $p \in [1,\infty]$ and every $\alpha > 0$.
Our main result shows that this upper bound is tight for all finite \(p \ge \log_2 3\) and for \(p=\infty\).

\begin{theorem}
The following holds:
\label{thm:ldlp-is-pi2-complete-lp}
\begin{itemize}
    \item \label{item:finite}For every constant $p \in [\log_2 3, \infty)$ and every constant $\alpha\in(2^{-1/p},1)$, $(p, \alpha)$-$\LDLP$ is $\Pi_2^p$-complete under deterministic polynomial-time many-one reductions.
    \item \label{item:infty}For every constant $\alpha\in(1/2,1)$, $(\infty, \alpha)$-$\LDLP$ is $\Pi_2^p$-complete under deterministic polynomial-time many-one reductions.
\end{itemize}
\end{theorem}
Locally dense lattices are central ingredients in hardness proofs for \(\SVP_p\) \cite{doi:10.1137/S0097539700373039/Mic01,journals/toc/Micciancio12,bennett_et_al:LIPIcs.APPROX/RANDOM.2023.37}.
Although locally dense lattices can be constructed explicitly, as shown by Wan's deterministic construction of coset-locally dense lattices \cite{wan2026}, which underlies his deterministic $\NP$-hardness result for \(\SVP_p\) and is also used in our proof of \Cref{thm:ldlp-is-pi2-complete-lp}, \Cref{thm:ldlp-is-pi2-complete-lp} shows that verification remains computationally intractable.
More precisely, given only a candidate tuple \((\mathbf A,\vec s,\ell,\mathbf T)\), deciding whether it satisfies the two conditions defining LDLP is hard in general, unless the polynomial hierarchy collapses.
Thus, the theorem separates explicit construction from verification: deterministic algorithms may output locally dense lattices, while verifying from the tuple alone that an arbitrary candidate satisfies the definition remains hard.
Also, Wan's deterministic construction is what makes our proof of the \(\Pi_2^p\)-hardness of LDLP possible.

The two formulations of locally dense lattices described above lead to two natural promise problem variants of LDLP. 
In the coefficient variant, the preimages under \(\mathbf{T}\) are coefficient vectors; in the coset variant, they are vectors in the shifted coset.
We define these promise problems, denoted \(\CoeffLDLP\) and \(\CosetLDLP\), in \Cref{sec:equivalence}; see \Cref{def:promise-coset-ldlp,def:var-ldlp}.

Our second result shows that these two promise problems are mutually reducible
in deterministic polynomial time.

\begin{theorem}[Informal; see \Cref{thm:ldlp-variant-reductions}]
\label{thm:ldlp-variant-reductions-intr}
The following holds (up to a small change in parameters):
\begin{itemize}
    \item There is a deterministic polynomial-time many-one reduction from $\CoeffLDLP$ to $\CosetLDLP$.
    \item There is a deterministic polynomial-time many-one reduction from $\CosetLDLP$ to $\CoeffLDLP$.
\end{itemize}
Moreover, in both reductions, the hypercube dimension \(k\) is preserved.
\end{theorem}

Applying these reductions to YES instances yields the following robustness statement for locally dense lattices themselves.
For technical reasons, the passage from coefficient-locally dense lattices to coset-locally dense lattices may require an arbitrarily small additive increase in the allowed radius, whereas the reverse direction preserves the radius.
\begin{corollary}
\label{thm:among-ldls}
For any $p \geq 1$, $\alpha > 0$, $\beta > 0$ and $k \geq 1$, the following holds:
\begin{itemize}
    \item There exists a deterministic polynomial-time algorithm that, given a $(p, \alpha, k)$-coefficient-locally dense lattice as input, outputs a $(p, \alpha + \beta, k)$-coset-locally dense lattice.
    \item There exists a deterministic polynomial-time algorithm that, given a $(p, \alpha, k)$-coset-locally dense lattice as input, outputs a $(p, \alpha, k)$-coefficient-locally dense lattice.
\end{itemize}
\end{corollary}
Therefore, combined with Wan's deterministic construction of coset-locally dense lattices, this yields a deterministic construction of coefficient-locally dense lattices.

\subsection{Techniques}
\paragraph*{Hardness of LDLP.}
The proof of \Cref{thm:ldlp-is-pi2-complete-lp} reduces from
\(\forall\exists\)~1-in-3-SAT.  
The main difficulty in the proof is to make the short coset vectors in an LDLP instance play the role of Boolean assignments.
In the reduction from
\[
        \forall \vec{u}\in\{0,1\}^k\ \exists \vec{v}\in\{0,1\}^t:
        \psi(\vec{u},\vec{v}),
\]
we need two properties.
First, for completeness, every assignment should appear as the projection of some short vector in one shifted coset.
Second, for soundness, every sufficiently short vector in that coset should itself be binary: otherwise a short non-Boolean vector might satisfy the properties of LDLP without corresponding to any truth assignment.
The first property can be satisfied by using a locally dense lattice. On the other hand, the second property does not follow from the properties of a general locally dense lattice.
In order to satisfy the second property, we use Wan's deterministic locally dense lattice \cite{wan2026} to realize all Boolean assignments \((\vec{u},\vec{v})\) as projections of short coset vectors.
Moreover, we observe that, for \(p\geq \log_2 3\), every sufficiently short vector in the relevant coset is a binary vector. 
Thus short coset vectors can be interpreted as Boolean assignments.

We then encode the 1-in-3 constraints by a linear system \(\mathbf{C}\vec{z}=\vec{b}\) with a constant gap: for Boolean \(\vec{z}\), the system holds exactly when \(\vec{z}\) satisfies all clauses, and otherwise some coordinate of \(\mathbf{C}\vec{z}-\vec{b}\) has absolute value at least \(2\).
The final coset vector has the form
\[
        \begin{pmatrix}
        M(\mathbf{C}\vec{x}-\vec{b})\\
        \vec{x}
        \end{pmatrix}.
\]
Hence the upper block vanishes exactly for satisfying assignments, while any
unsatisfied clause makes the vector too long.  The target hypercube vector fixes
the universal assignment \(\vec{u}\), and local density holds exactly when some
existential assignment \(\vec{v}\) satisfies \(\psi(\vec{u},\vec{v})\).

For \(\ell_\infty\), we first prove hardness for CoeffLDLP.
The top block \(2\mathbf{I}\) together with target \(\mathbf{1}\) forces every
short integer coefficient vector to be Boolean, since each coordinate
\(2z_i-1\) is an odd integer of absolute value less than \(2\).
The remaining rows encode the same 1-in-3 constraints. 
Combining this CoeffLDLP hardness with \Cref{thm:ldlp-variant-reductions-intr} yields Item 2 of \Cref{thm:ldlp-is-pi2-complete-lp}.

\paragraph*{Equivalence of the two formulations.}
The proof of \Cref{thm:ldlp-variant-reductions} and \Cref{thm:among-ldls} proceeds by mutual polynomial-time reductions between the two LDLP variants.
The difficulty is that passing between coefficient vectors and coset vectors does not preserve the auxiliary map \(\mathbf{T}\). 
For the reduction from \(\CoeffLDLP\) to \(\CosetLDLP\), we first prove an a priori bound of polynomial bit length on the relevant coefficient vectors, and then encode a coefficient vector \(\vec{x}\) as
\[
        \begin{pmatrix}
        M(\mathbf{A}\vec{x}-\vec{s})\\
        \vec{x}
        \end{pmatrix},
\]
so that the lower block keeps the information needed to apply the original
map \(\mathbf{T}\), while the scaling \(M\) absorbs its norm contribution into
the additive slack \(\beta\).  In the reverse direction, we add coefficient
variables \(\vec{y}\) intended to equal the Boolean image of a coset vector and
use the residual
\[
        \begin{pmatrix}
        \vec{s}+\mathbf{A}\vec{z}\\
        C(\mathbf{T}(\vec{s}+\mathbf{A}\vec{z})-\vec{y})
        \end{pmatrix}.
\]
The second block penalizes any mismatch, so short solutions must satisfy
\(\vec{y}=\mathbf{T}(\vec{s}+\mathbf{A}\vec{z})\), after which the new
coefficient map outputs \(\vec{y}\) directly.  The two reductions preserve
\(k\), and hence establish the claimed equivalence.

\subsection{Related work}
Several hardness results for \(\SVP_p\) do not use the locally dense lattices discussed above, including Khot's results~\cite{journals/jcss/Kho06} in high \(\ell_p\) norms, the results for \(\SVP_\infty\) of van Emde Boas \cite{1571698600177703168/vanEmdeBoas} and Dinur \cite{DINUR200255}, and recent deterministic hardness results \cite{hecht2025deterministichardnessapproximationuniquesvp, hair2025svppnphardp,hittmeir2026finegraineddeterministichardnessshortest, hair2026finite}.

Variants of locally dense lattices also appear in several related settings where the formal gadget differs slightly from the one considered in this paper.
In particular, Khot's reduction~\cite{10.1145/1089023.1089027/Kho05}, the results of Haviv and Regev~\cite{journals/toc/HR12}, and later work on fine-grained hardness and parameterized complexity  \cite{bennett_et_al:LIPIcs.CCC.2020.36/BP20,conf/stoc/AggarwalS18,conf/innovations/BennettPT22,conf/stoc/BennettCGR23} use sparsification or related randomized filtering rather than a linear map onto a Boolean hypercube.

There is also a close coding-theoretic analogue: locally dense codes underlie randomized and deterministic hardness results for Minimum Distance Problem (MDP) \cite{1159759/DMS03,6268342/CW12,10.1109/CCC.2014.17/Mic14,6868217/AK14}, and recent work on sparse vector problems in codes, subspaces, and lattices uses related constructions, with sparsity as the relevant parameter \cite{bhattiprolu2025inapproximabilityfindingsparsevectors}.

The quantifier structure of \(\LDLP\) is also related to the Covering Radius Problem (CRP).
The CRP asks whether every point of the ambient space is close to a lattice or a code, whereas \(\LDLP\) asks whether every Boolean label has a short preimage in one shifted coset.
Guruswami, Micciancio, and Regev \cite{journals/cc/GuruswamiMR2005} studied the CRP for lattices and codes, and proved \(\Pi_2^p\)-hardness of approximation for the code version for some constant factor.
Haviv and Regev~\cite{journals/ChicagoTCS/HavivR12} later proved \(\Pi_2^p\)-hardness of approximation for the lattice version in large \(\ell_p\) norms.
Another related problem at the second level of the polynomial hierarchy is Linear Discrepancy, for which Manurangsi~\cite{MANURANGSI2021106164} proved \(\Pi_2^p\)-hardness of approximation via a reduction from a quantified SAT problem.
More recently, Bennett and Ly~\cite{bennett2026hardnessbinarycoveringradius} studied the Binary Covering Radius Problem, a promise variant that connects CRP and Linear Discrepancy, and proved new hardness results for large \(\ell_p\) norms and for \(\ell_\infty\).

\subsection{Open problems}
Our main result shows that \((p,\alpha)\)-\(\LDLP\) is \(\Pi_2^p\)-complete for every finite \(p \ge \log_2 3\) and every \(\alpha \in (2^{-1/p},1)\).
It remains open whether the same holds for every \(1 \le p < \log_2 3\) and every \(\alpha \in (2^{-1/p},1)\).
The obstacle in our proof is the step that forces sufficiently short vectors in Wan's shifted coset to be Boolean.
Regev and Rosen~\cite{conf/stoc/RegevR06} gave randomized reductions between lattice problems in different \(\ell_p\) norms using embeddings between normed spaces.
Since our hardness result includes the Euclidean norm, it is natural to ask whether their reductions can be combined with our construction to obtain hardness of \(\LDLP\) under randomized reductions for every \(p \ge 1\).

There is an analogous recognition problem for codes.
Locally dense codes have been used in hardness results for MDP~\cite{1159759/DMS03,6268342/CW12,10.1109/CCC.2014.17/Mic14,6868217/AK14}.
It would be interesting to determine whether recognizing locally dense codes is also \(\Pi_2^p\)-complete.

The quantifier structure of \(\LDLP\) is related to the CRP.
Haviv and Regev~\cite{journals/ChicagoTCS/HavivR12} proved \(\Pi_2^p\)-hardness of approximation for the lattice version in large \(\ell_p\) norms.
An interesting question is whether similar techniques can prove \(\Pi_2^p\)-hardness for the Euclidean lattice version.
\section{Preliminaries}
\label{sec:prelims}
\subsection{Notations}
For a positive integer $n$, let $[n]:=\{1,\ldots,n\}$. For a prime
$q$, let $\mathbb{F}_q:=\mathbb{Z}/q\mathbb{Z}$ be the finite field
with $q$ elements. We write $\mathbf{I}_r$ for the $r\times r$
identity matrix, and $\mathbf{0}_{a\times b}$ or
$\mathbf{O}_{a\times b}$ for the $a\times b$ zero matrix. When the
dimensions are clear, we omit the subscripts. For positive integers
$r\leq n$, let
\[
    \mathbf{P}_r
    :=
    \begin{pmatrix}
        \mathbf{I}_r & \mathbf{0}_{r\times(n-r)}
    \end{pmatrix}
    \in\mathbb{Z}^{r\times n}
\]
denote the projection onto the first $r$ coordinates, where the
ambient dimension $n$ will always be clear from context.

A lattice $\mathcal{L}$ is the set of integer linear combinations of $n$ linearly independent vectors $\vec{b}_1,\ldots,\vec{b}_n \in \mathbb{R}^m$, i.e.,
\begin{align*}
    \mathcal{L} = \mathcal{L}(\vec{b}_1,\ldots,\vec{b}_n) := \left\{\sum^n_{i=1} a_i\vec{b}_i\ :\ a_1,\ldots,a_n\in \mathbb{Z}\right\}.
\end{align*}
The set of vectors $\vec{b}_1,\ldots,\vec{b}_n$ is called a $basis$ of the lattice $\mathcal{L}(\vec{b}_1,\ldots,\vec{b}_n)$, and $n$ is called the $rank$ of the lattice. The basis can be represented by the matrix $\mathbf{B} := (\vec{b}_1,\ldots,\vec{b}_n)$, and the lattice generated by $\mathbf{B}$ is denoted as $\mathcal{L}(\mathbf{B})$.

For any $p \in [1, \infty)$ and a vector $\vec{x}\in \mathbb{R}^n$, the length of $\vec{x}$ in the $\ell_p$ norm is defined as follows:
\begin{align*}
    \lVert\vec{x}\rVert_p := (|x_1|^p + |x_2|^p + \cdots + |x_n|^p)^{1/p}.
\end{align*}
For $p=\infty$, we denote $\lVert\vec{x}\rVert_\infty = \max_{i\in[n]} |x_i|$.
We denote the length of a shortest non-zero vector in a lattice $\mathcal{L}$ with basis $\mathbf{B}\in\mathbb{R}^{m\times n}$ in the $\ell_p$ norm as:
\begin{align*}
    \lambda^{(p)}_1(\mathcal{L}(\mathbf{B})) := \underset{\vec{x}\in\mathcal{L}(\mathbf{B})\backslash\{0\}}{\min} \lVert\vec{x}\rVert_p.
\end{align*}
For a lattice $\mathcal{L}$ with basis $\mathbf{B} \in \mathbb{R}^{m\times n}$, and a vector $\vec{x}\in\mathbb{R}^m$, let $\dist_p(\mathcal{L}(\mathbf{B}), \vec{x}) := \min_{\vec{y}\in \mathcal{L}(\mathbf{B})}\lVert\vec{x} - \vec{y}\rVert_p$ be the minimum distance in the $\ell_p$ norm between $\vec{x}$ and a lattice vector of $\mathcal{L}$. For a real number $r\in\mathbb{R}$ and a vector $\vec{s}\in\mathbb{R}^n$, let $\mathcal{B}^n_p(\vec{s},r) := \{\vec{x}\in\mathbb{R}^n : \lVert\vec{x}-\vec{s}\rVert_p\leq r\}$ be the $n$-dimensional closed $\ell_p$ ball of radius $r$ centered at $\vec{s}$. Similarly, let $\mathcal{B}^n_p(r) := \mathcal{B}^n_p(\vec{0}, r) = \{\vec{x}\in\mathbb{R}^n : \lVert\vec{x}\rVert_p\leq r\}$ be the $n$-dimensional closed $\ell_p$ ball of radius $r$ centered at the origin. When the ambient dimension is clear, we omit the superscript $n$.
\subsection{Locally Dense Lattices}
A \emph{locally dense lattice} for relative distance $\alpha > 0$ consists of a lattice $\mathcal{L}$, an integer vector $\vec{s}$ and a linear transformation matrix $\mathbf{T}$ that maps to a $k$-dimensional hypercube, where there exist at least exponentially many lattice vectors (in the lattice rank) in the ball $\mathcal{B}^m_p(\vec{s}, \alpha\cdot\lambda^{(p)}_1(\mathcal{L}))$. In other words, the coset $\vec{s} + \mathcal{L}$, which the lattice $\mathcal{L}$ is shifted by the vector $\vec{s}$ contains many relatively short vectors compared to the length of the shortest non-zero lattice vector in $\mathcal{L}$. The linear transformation matrix $\mathbf{T}$ maps short vectors in the coset to a $k$-dimensional hypercube.
\begin{definition}[$(p, \alpha, k)$-coset-Locally Dense Lattice, \cite{journals/toc/Micciancio12,bennett_et_al:LIPIcs.APPROX/RANDOM.2023.37}]
    \label{def:ldl}
    For $p \in [1, \infty]$, a real number $\alpha > 0$ and a positive integer $k>0$, a $(p, \alpha, k)$-\emph{coset-locally dense lattice} consists of a lattice basis $\mathbf{A}\in\mathbb{Z}^{m\times n}$ of rank $n$ and dimension $m$, an integer vector $\vec{s} \in \mathbb{Z}^m$, a distance threshold $\ell>0$, and a linear transformation matrix $\mathbf{T}\in\mathbb{Z}^{k\times m}$, where:
    \begin{enumerate}
        \item \label{item:min-dist} $\lambda^{(p)}_1(\mathcal{L}(\mathbf{A})) \geq \ell$ and
        \item \label{item:conv-hypercube} $\{0,1\}^k \subseteq \{\mathbf{T}\vec{v} : \vec{v}\in (\vec{s} + \mathcal{L}(\mathbf{A}))\cap \mathcal{B}^m_p(\alpha\cdot\ell)\}$.
    \end{enumerate}
\end{definition}

The following is a natural variant of the $(p,\alpha,k)$-coset-locally dense lattice problem, in which the preimage under a linear map to the hypercube is specified by coefficient vectors with respect to the lattice basis, rather than by vectors in the coset itself.
We note that a locally dense lattice that maps coefficient vectors to the hypercube was introduced by Micciancio \cite{doi:10.1137/S0097539700373039/Mic01} to show the $\NP$-hardness of $\gamma$-$\GapSVP_p$ under randomized reductions.
\begin{definition}[$(p, \alpha, k)$-coefficient-Locally Dense Lattice, \cite{doi:10.1137/S0097539700373039/Mic01}]
    \label{def:ldl-mic}
    For $p \in [1, \infty]$, a real number $\alpha > 0$ and a positive integer $k>0$, a $(p, \alpha, k)$-\emph{coefficient-locally dense lattice} consists of a lattice basis $\mathbf{A}\in\mathbb{Z}^{m\times n}$ of rank $n$ and dimension $m$, an integer vector $\vec{s} \in \mathbb{Z}^m$, a distance threshold $\ell>0$, and a linear transformation matrix $\mathbf{T}\in\mathbb{Z}^{k\times n}$, where:
    \begin{enumerate}
        \item \label{item:min-dist-mic} $\lambda^{(p)}_1(\mathcal{L}(\mathbf{A})) \geq \ell$ and
        \item \label{item:conv-hypercube-mic} $\{0,1\}^k \subseteq \{\mathbf{T}\vec{v} : \vec{v}\in \mathbb{Z}^n, \|\mathbf{A}\vec{v} - \vec{s}\|_p\leq \alpha\cdot\ell\}$.
    \end{enumerate}
\end{definition}

The following theorem gives a deterministic polynomial-time algorithm, due to Wan \cite{wan2026}, that outputs a (coset-)locally dense lattice. In addition to the locally dense lattice property, this construction also has the property that every short vector in the coset is a binary vector. Since this property is not stated explicitly in the paper by Wan, we provide a proof below for completeness (see \Cref{appendix:ldl}).
\begin{theorem}[\cite{wan2026}]
\label{lem:binary-rs-gadget-lp}
Fix constants $p \in [1,\infty)$ and $\alpha \in (2^{-1/p},1)$.
There exists a deterministic polynomial-time algorithm that, given a positive integer $r$,
outputs a lattice basis $\mathbf{A} \in \mathbb{Z}^{N\times d}$, an integer vector
$\vec{s} \in \mathbb{Z}^N$, a positive integer $\ell$, and the projection matrix
$\mathbf{P}_r \in \mathbb{Z}^{r\times N}$ onto the first $r$ coordinates such that:
\begin{enumerate}
    \item \label{item:ldl-1} $\lambda_1^{(p)}(\mathcal{L}(\mathbf{A})) \ge \ell$,
    \item \label{item:ldl-2} $\{0,1\}^r \subseteq \{\mathbf{P}_r\vec{x} : \vec{x} \in (\vec{s} + \mathcal{L}(\mathbf{A})) \cap \mathcal{B}_p^N(\alpha\ell)\}$.
\end{enumerate}
Moreover,
\label{item:ldl-3}
if $p \ge \log_2 3$, then every vector in $(\vec{s} + \mathcal{L}(\mathbf{A})) \cap \mathcal{B}_p^N(\alpha\ell)$ belongs to $\{0,1\}^N$.

In all cases, the total bit length of $(\mathbf{A},\vec{s},\ell,\mathbf{P}_r)$ is bounded by $r^{O(1)}$.
\end{theorem}

\subsection{Locally Dense Lattice Problem}
We define the \emph{Locally Dense Lattice Problem} (LDLP), which is the decision problem of determining whether a given input specifies a (coset-)locally dense lattice. The input consists of a lattice basis, a shift vector, a positive integer, and a linear transformation matrix.
\begin{definition}[$(p, \alpha)$-LDLP]
\label{def:ldlp}
For $p \in [1, \infty]$, $\alpha > 0$, the decision problem $(p,\alpha)$-\emph{Locally Dense Lattice Problem} \emph{($(p,\alpha)$-$\LDLP$)} is defined as follows. Given a tuple $(\mathbf{A}, \vec{s}, \ell, \mathbf{T})$, where $\mathbf{A} \in \mathbb{Z}^{m \times n}$ is a lattice basis, $\vec{s} \in \mathbb{Z}^m$ is an integer vector, $\ell > 0$ is a distance threshold, and $\mathbf{T} \in \mathbb{Z}^{k \times m}$ is a linear transformation matrix, determine whether it satisfies the following conditions:
\begin{enumerate}
    \item \label{item:ldlp-min-dist} $\lambda_1^{(p)}(\mathcal{L}(\mathbf{A})) \geq \ell$,
    \item \label{item:ldlp-conv-hypercube} $\{0,1\}^k \subseteq \{\mathbf{T} \vec{v} : \vec{v} \in (\vec{s} + \mathcal{L}(\mathbf{A})) \cap \mathcal{B}_p^m(\alpha \cdot \ell)\}$.
\end{enumerate}
Here and below, \(k\) denotes the number of rows of \(\mathbf{T}\).
\end{definition}

Note that if $(p, \alpha)$-$\LDLP$ instance $(\mathbf{A}, \vec{s}, \ell, \mathbf{T})$ satisfies \Cref{item:ldlp-min-dist} and \Cref{item:ldlp-conv-hypercube} of \Cref{def:ldlp}, and $\mathbf{T}$ has $k$ rows, then $(\mathbf{A}, \vec{s}, \ell, \mathbf{T})$ is $(p, \alpha, k)$-coset-locally dense lattice.

As a fundamental property of $(p, \alpha)$-$\LDLP$, $(p, \alpha)$-$\LDLP$ is in $\Pi^p_2$.
\begin{proposition}
\label{prop:pi2p}
    For any fixed $p\in[1,\infty]$, $\alpha > 0$, $(p, \alpha)$-$\LDLP$ is in $\Pi_2^p$.
\end{proposition}
\begin{proof}
Fix $p \in [1,\infty]$, $\alpha > 0$. Let $(\mathbf{A}, \vec{s}, \ell, \mathbf{T})$ be an input to $(p, \alpha)$-$\LDLP$.
By \Cref{def:ldlp}, $(\mathbf{A}, \vec{s}, \ell, \mathbf{T})$ belongs to $(p, \alpha)$-$\LDLP$ if and only if all items of the definition hold.

First, consider \Cref{item:ldlp-min-dist}. Its complement is the statement that there exists a nonzero vector $\vec{w}\in \mathcal{L}(\mathbf{A})$ such that $\|\vec{w}\|_p < \ell$.
Such a vector $\vec{w}$ is a polynomial-size witness, because $\|\vec{w}\|_p<\ell$ implies $|w_i|<\ell$ for every coordinate $i$.
We can verify in polynomial time that $\vec{w}\neq \vec{0}$, that $\vec{w}\in \mathcal{L}(\mathbf{A})$ by standard Hermite-normal-form or Smith-normal-form algorithms, and that $\|\vec{w}\|_p < \ell$.
Hence the complement of \Cref{item:ldlp-min-dist} is in $\NP$, and therefore Item~1 is in $\coNP \subseteq \Pi_2^p$.

Next, \Cref{item:ldlp-conv-hypercube} holds if and only if
\begin{align*}
    \forall \vec{x}\in\{0,1\}^{k},\ \exists \vec{v}\in\mathbb{Z}^{m}
    \text{ such that }
    \vec{v}\in (\vec{s} + \mathcal{L}(\mathbf{A}))\cap \mathcal{B}_p^m(\alpha\ell)
    \text{ and } \mathbf{T}\vec{v}=\vec{x}.
\end{align*}
For any $\vec{x}\in\{0,1\}^k$ and $\vec{v}\in\mathbb{Z}^m$, we can verify in polynomial time whether $\vec{v}-\vec{s}\in \mathcal{L}(\mathbf{A})$, whether $\vec{v}\in \mathcal{B}_p^m(\alpha\ell)$, and whether $\mathbf{T}\vec{v}=\vec{x}$.
Thus, Item~2 is in $\Pi_2^p$.

Since $\Pi_2^p$ contains $\coNP$, and is closed under conjunction, it follows that $(p, \alpha)$-$\LDLP$ is in $\Pi_2^p$.
\end{proof}

\section{Equivalence of Locally Dense Lattices}
\label{sec:equivalence}
In this section, by considering the computational complexity of LDLP, we prove \Cref{thm:among-ldls}.

For convenience, we define a natural variant of LDLP, called CosetLDLP, as a promise problem. In this problem, \Cref{item:min-dist} of a locally dense lattice, namely that $\lambda_1^{(p)}(\mathcal{L}(\mathbf{A}))$ satisfies the required bound, is assumed as a promise, and the task is to decide whether \Cref{item:conv-hypercube} of a coset-locally dense lattice is satisfied.
\begin{definition}[$(p, \alpha)$-CosetLDLP]
\label{def:promise-coset-ldlp}
For $p \in [1, \infty]$ and $\alpha > 0$, the promise problem \emph{$(p, \alpha)$-$\CosetLDLP$} is defined as follows. Instances are tuples $(\mathbf{A}, \vec{s}, \ell, \mathbf{T})$, where $\mathbf{A} \in \mathbb{Z}^{m \times n}$ is a lattice basis, $\vec{s} \in \mathbb{Z}^m$ is an integer vector, $\ell > 0$ is a distance threshold satisfying $\lambda^{(p)}_1(\mathcal{L}(\mathbf{A})) \geq \ell$, and $\mathbf{T} \in \mathbb{Z}^{k \times m}$ is a linear transformation matrix. Then:
\begin{enumerate}
    \item $(\mathbf{A}, \vec{s}, \ell, \mathbf{T})$ is a YES instance if $\{0,1\}^k \subseteq \{\mathbf{T} \vec{v} : \vec{v} \in (\vec{s} + \mathcal{L}(\mathbf{A})) \cap \mathcal{B}_p^m(\alpha \cdot \ell)\}$,
    \item $(\mathbf{A}, \vec{s}, \ell, \mathbf{T})$ is a NO instance if $\{0,1\}^k \nsubseteq \{\mathbf{T} \vec{v} : \vec{v} \in (\vec{s} + \mathcal{L}(\mathbf{A})) \cap \mathcal{B}_p^m(\alpha \cdot \ell)\}$.
\end{enumerate}
Here and below, \(k\) denotes the number of rows of \(\mathbf{T}\).
\end{definition}
Note that there is a trivial reduction from $(p, \alpha)$-$\CosetLDLP$ to $(p, \alpha)$-$\LDLP$.

We define the analogous problem for coefficient-locally dense lattices as well. In this problem, the task is to decide whether condition \Cref{item:conv-hypercube-mic} of a coefficient-locally dense lattice is satisfied.
\begin{definition}[$(p, \alpha, \beta)$-CoeffLDLP]
\label{def:var-ldlp}
For $p \in [1, \infty]$, $\alpha > 0$ and $\beta \geq 0$, the promise problem $(p, \alpha, \beta)$-\emph{Coeff Locally Dense Lattice Problem} \emph{($(p, \alpha, \beta)$-$\CoeffLDLP$)} is defined as follows. Instances are tuples $(\mathbf{A}, \vec{s}, \ell, \mathbf{T})$, where $\mathbf{A} \in \mathbb{Z}^{m \times n}$ is a lattice basis, $\vec{s} \in \mathbb{Z}^m$ is an integer vector, $\ell > 0$ is a distance threshold satisfying $\lambda^{(p)}_1(\mathcal{L}(\mathbf{A})) \geq \ell$, and $\mathbf{T} \in \mathbb{Z}^{k \times n}$ is a linear transformation matrix. Then:
\begin{enumerate}
    \item $(\mathbf{A}, \vec{s}, \ell, \mathbf{T})$ is a YES instance if $\{0,1\}^{k} \subseteq  \{\mathbf{T}\vec{v} : \vec{v}\in \mathbb{Z}^n, \lVert\mathbf{A}\vec{v} - \vec{s}\rVert_p\leq \alpha\cdot \ell\}$,
    \item $(\mathbf{A}, \vec{s}, \ell, \mathbf{T})$ is a NO instance if $\{0,1\}^{k} \nsubseteq  \{\mathbf{T}\vec{v} : \vec{v}\in\mathbb{Z}^n, \lVert\mathbf{A}\vec{v} - \vec{s}\rVert_p\leq (\alpha + \beta)\cdot \ell\}$.
\end{enumerate}
Here and below, \(k\) denotes the number of rows of \(\mathbf{T}\).
\end{definition}

We first prove that the polynomial-time reduction from $(p,\alpha,\beta)$-$\CoeffLDLP$ to $(p, \alpha + \beta)$-$\CosetLDLP$.

The following fact provides an upper bound on the length of coefficient vectors corresponding to lattice vectors that lie within distance $\alpha \ell$ from the center $\vec{s}$.
\begin{lemma}
\label{fact:upper-vector}
Let $\mathbf{A}\in\mathbb{Z}^{m\times n}$ have column rank $n$.
Fix $\vec{s}\in\mathbb{R}^m$ and $\alpha,\ell\geq 0$.
Let $p\in(0,\infty]$. For any $\vec{x}\in\mathbb{Z}^n$ satisfying $\lVert \mathbf{A}\vec{x}-\vec{s}\rVert_p \leq \alpha\ell$,
we have
\begin{align*}
\lVert \vec{x}\rVert_\infty \leq (\sqrt{n}\,H)^n,
\end{align*}
where $H := \max\bigl(\lVert \mathbf{A}\rVert_\infty,\ \lVert \vec{s}\rVert_\infty+\lceil \alpha\ell\rceil\bigr)$,
and $\lVert \mathbf{A}\rVert_\infty := \max_{i\in[m],\,j\in[n]} |a_{ij}|$. Moreover, for $p\in(0,\infty)$ we also have
\begin{align*}
\lVert \vec{x}\rVert_p \leq n^{1/p}(\sqrt{n}\,H)^n.
\end{align*}
\end{lemma}
\begin{proof}
Since $\mathbf{A} \in \mathbb{Z}^{m \times n}$ has column rank $n$, there exist
$n$ rows of $\mathbf{A}$ that form an invertible $n \times n$ submatrix.
Let $\mathbf{B} \in \mathbb{Z}^{n \times n}$ be such a submatrix.

Set $\vec{y} := \mathbf{A}\vec{x} \in \mathbb{Z}^m$.
Since $\|\mathbf{A}\vec{x} - \vec{s}\|_p \leq \alpha \ell$, every coordinate of
$\mathbf{A}\vec{x} - \vec{s}$ has absolute value at most $\alpha \ell$.
Hence each coordinate of $\vec{y}$ has absolute value at most
$\|\vec{s}\|_\infty + \alpha \ell \leq \|\vec{s}\|_\infty + \lceil \alpha \ell \rceil \leq H$.
Also, every entry of $\mathbf{A}$ has absolute value at most $\|\mathbf{A}\|_\infty \leq H$.
Therefore every entry of $\mathbf{B}$ is also bounded in absolute value by $H$.
Let $\vec{y}^{\,\prime} \in \mathbb{Z}^n$ be the vector consisting of the coordinates of $\vec{y}$
corresponding to the rows used to define $\mathbf{B}$.
Then $\mathbf{B}\vec{x} = \vec{y}^{\,\prime}$.
For each $j \in \{1,\dots,n\}$, let $\mathbf{B}_j$ be the matrix obtained from $\mathbf{B}$
by replacing its $j$-th column with $\vec{y}^{\,\prime}$.
By Cramer's rule,
\[
x_j = \frac{\det(\mathbf{B}_j)}{\det(\mathbf{B})}.
\]
Since $\mathbf{B}$ is an invertible integer matrix, $\det(\mathbf{B})$ is a nonzero integer.
Hence $|\det(\mathbf{B})| \geq 1$.

Next, every entry of $\mathbf{B}_j$ has absolute value at most $H$, so every column of
$\mathbf{B}_j$ has Euclidean norm at most $\sqrt{n}\,H$.
Therefore, by Hadamard's inequality, $|\det(\mathbf{B}_j)| \leq (\sqrt{n}\,H)^n$.
It follows that
\[
|x_j|
= \left|\frac{\det(\mathbf{B}_j)}{\det(\mathbf{B})}\right|
\leq (\sqrt{n}\,H)^n
\qquad \text{for all } j \in \{1,\dots,n\}.
\]
Hence $\|\vec{x}\|_\infty \leq (\sqrt{n}\,H)^n$.
If $p \in (0,\infty)$, then
$\|\vec{x}\|_p \leq n^{1/p}\|\vec{x}\|_\infty$,
and therefore $\|\vec{x}\|_p \leq n^{1/p}(\sqrt{n}\,H)^n$.
\end{proof}

Next, we give a reduction from $(p,\alpha,\beta)$-$\CoeffLDLP$ to $(p, \alpha + \beta)$-$\CosetLDLP$.
\begin{proposition}
\label{prop:coefficient-to-coset}
For any $p \in [1, \infty]$, $\alpha > 0$ and $\beta > 0$, there exists a deterministic polynomial-time many-one reduction from $(p, \alpha, \beta)$-$\CoeffLDLP$ to $(p, \alpha + \beta)$-$\CosetLDLP$.
Moreover, if the input transformation matrix has $k$ rows, then so does the output transformation matrix.
\end{proposition}
\begin{proof}
Let $(\mathbf{A}, \vec{s}, \ell, \mathbf{T})$ be an instance of $(p,\alpha,\beta)$-$\CoeffLDLP$, where $\mathbf{A} \in \mathbb{Z}^{m \times n}$ is a lattice basis, $\vec{s} \in \mathbb{Z}^m$ is an integer vector, $\ell > 0$ is a distance threshold, and $\mathbf{T} \in \mathbb{Z}^{k \times n}$ is a linear transformation matrix.
The reduction algorithm outputs
\begin{align*}
    \mathbf{A}' &:=
    \begin{pmatrix}
        M\mathbf{A} \\
        \mathbf{I}_n
    \end{pmatrix}
    \in \mathbb{Z}^{(m+n) \times n}
    \qquad \text{with }
    M :=
    \begin{cases}
        \left\lceil n^{1/p}(\sqrt{n}\,H)^n / (\beta \ell) \right\rceil,
        & \text{if } p \in [1,\infty),\\
        \left\lceil (\sqrt{n}\,H)^n / (\beta \ell) \right\rceil,
        & \text{if } p = \infty,
    \end{cases}
    \\
    \vec{s}' &:=
    \begin{pmatrix}
        -M \vec{s} \\
        \vec{0}_n
    \end{pmatrix}
    \in \mathbb{Z}^{m+n},
    \qquad
    \ell' := M \ell > 0,
    \\
    \mathbf{T}' &:=
    \begin{pmatrix}
        \mathbf{O}_{k\times m} & \mathbf{T}
    \end{pmatrix}
    \in \mathbb{Z}^{k\times(m+n)},
\end{align*}
where
$H := \max\bigl(\lVert \mathbf{A}\rVert_\infty,\ \lVert \vec{s}\rVert_\infty + \lceil \alpha\ell\rceil\bigr)$, $\lVert \mathbf{A}\rVert_\infty := \max_{i\in[m],\,j\in[n]} |a_{ij}|$, $\mathbf{O}_{k\times m}$
is the $k\times m$ zero matrix, $\mathbf{I}_n$ is the $n\times n$ identity matrix, and $\vec{0}_n$ is the $n$-dimensional zero vector.

It is clear that the reduction runs in polynomial time.

Moreover, $(\mathbf{A}',\vec{s}',\ell', \mathbf{T}')$ satisfies \Cref{item:ldlp-min-dist} of \Cref{def:ldlp}: $\lambda_1^{(p)}(\mathcal{L}(\mathbf{A}')) \ge \ell'$.
Indeed, for every nonzero $\vec{x}\in \mathbb{Z}^n$,
\begin{align*}
    \lVert \mathbf{A}'\vec{x} \rVert_p
    =
    \left\lVert
        \begin{pmatrix}
            M\mathbf{A}\vec{x} \\
            \vec{x}
        \end{pmatrix}
    \right\rVert_p
    \ge M \lVert \mathbf{A}\vec{x} \rVert_p
    \ge M \lambda_1^{(p)}(\mathcal{L}(\mathbf{A}))
    \ge M\ell
    = \ell',
\end{align*}
and hence $\lambda_1^{(p)}(\mathcal{L}(\mathbf{A}')) \ge \ell'$.

We first show the completeness.
Assume that $(\mathbf{A}, \vec{s}, \ell, \mathbf{T})$ is a YES instance of
$(p,\alpha,\beta)$-$\CoeffLDLP$.
Fix an arbitrary vector $\vec{y} \in \{0,1\}^k$ and choose
$\vec{x} \in \mathbb{Z}^n$ such that $\mathbf{T}\vec{x} = \vec{y}$ and
$\lVert \mathbf{A}\vec{x} - \vec{s} \rVert_p \le \alpha \ell$.
Define $\vec{v} := \vec{s}' + \mathbf{A}'\vec{x}$.
Then $\vec{v} \in \vec{s}' + \mathcal{L}(\mathbf{A}')$ and
$\mathbf{T}'\vec{v} = \vec{y}$.
Using the inequality
$\lVert (\vec{u},\vec{w}) \rVert_p \le \lVert \vec{u} \rVert_p + \lVert \vec{w} \rVert_p$
for all $p \in [1,\infty]$, we obtain
\begin{align*}
    \lVert \vec{v} \rVert_p
    &= \left\lVert
        \begin{pmatrix}
            M(\mathbf{A}\vec{x}-\vec{s}) \\
            \vec{x}
        \end{pmatrix}
    \right\rVert_p
    \le  M\lVert \mathbf{A}\vec{x}-\vec{s} \rVert_p + \lVert \vec{x} \rVert_p.
\end{align*}
By \Cref{fact:upper-vector} and the definition of $M$, we have
\begin{align*}
    \lVert \vec{v} \rVert_p
    \le M\alpha\ell + M\beta\ell
    = M(\alpha+\beta)\ell
    = (\alpha+\beta)\ell'.
\end{align*}
Since $\vec{y}$ was arbitrary, this shows that \Cref{item:ldlp-conv-hypercube} of \Cref{def:ldlp} holds for $(\mathbf{A}', \vec{s}', \ell', \mathbf{T}')$ with parameter $\alpha+\beta$.
Together with \Cref{item:ldlp-min-dist} of \Cref{def:ldlp}, we conclude that $(\mathbf{A}', \vec{s}', \ell', \mathbf{T}')$ is a YES instance of $(p,\alpha+\beta)\text{-}\CosetLDLP$.

Next, we show the soundness.
Assume that $(\mathbf{A}, \vec{s}, \ell, \mathbf{T})$ is a NO instance of $(p,\alpha,\beta)$-$\CoeffLDLP$.
Then there exists a vector $\vec{y} \in \{0,1\}^k$ such that for all $\vec{x} \in \mathbb{Z}^n$,
\begin{align*}
    \mathbf{T}\vec{x} = \vec{y}
    \ \Rightarrow\
    \lVert \mathbf{A}\vec{x} - \vec{s} \rVert_p > (\alpha+\beta)\ell.
\end{align*}
Fix any vector $\vec{v} = \vec{s}' + \mathbf{A}'\vec{x} \in \vec{s}' + \mathcal{L}(\mathbf{A}')$ for some $\vec{x} \in \mathbb{Z}^{n}$ such that $\mathbf{T}'\vec{v} = \vec{y}$.
Since $\mathbf{T}\vec{x} = \mathbf{T}'\vec{v} = \vec{y}$, it follows that $\lVert \mathbf{A}\vec{x} - \vec{s} \rVert_p > (\alpha+\beta)\ell$.
Then, we obtain
\begin{align*}
    \lVert \vec{v} \rVert_p
    &= \left\lVert
        \begin{pmatrix}
            M(\mathbf{A}\vec{x}-\vec{s}) \\
            \vec{x}
        \end{pmatrix}
    \right\rVert_p
    \ge \left\lVert M(\mathbf{A}\vec{x}-\vec{s}) \right\rVert_p
    > M(\alpha+\beta)\ell
    = (\alpha+\beta)\ell'.
\end{align*}
Therefore, for this vector $\vec{y}$, there is no $\vec{v} \in \vec{s}' + \mathcal{L}(\mathbf{A}')$ such that $\mathbf{T}'\vec{v} = \vec{y}$ and $\lVert \vec{v} \rVert_p \le (\alpha+\beta)\ell'$.
Hence \Cref{item:ldlp-conv-hypercube} of \Cref{def:ldlp} fails for $(\mathbf{A}', \vec{s}', \ell', \mathbf{T}')$, and so $ (\mathbf{A}', \vec{s}', \ell', \mathbf{T}')$ is a NO instance of $(p,\alpha+\beta)\text{-}\CosetLDLP$.
\end{proof}
We give a reduction from $(p,\alpha)$-$\CosetLDLP$ to $(p,\alpha,0)$-$\CoeffLDLP$.

\begin{proposition}
\label{prop:coset-to-coefficient}
For any $p \in [1,\infty]$ and $\alpha > 0$, there exists a deterministic polynomial-time many-one reduction from $(p,\alpha)$-$\CosetLDLP$ to $(p,\alpha,0)$-$\CoeffLDLP$.
Moreover, if the input transformation matrix has $k$ rows, then so does the output transformation matrix.
\end{proposition}
\begin{proof}
Let $(\mathbf{A}, \vec{s}, \ell, \mathbf{T})$ be an instance of $(p,\alpha)$-$\CosetLDLP$, where $\mathbf{A} \in \mathbb{Z}^{m \times n}$ is a lattice basis, $\vec{s} \in \mathbb{Z}^m$ is an integer vector, $\ell > 0$ is a distance threshold, and $\mathbf{T} \in \mathbb{Z}^{k \times m}$ is a linear transformation matrix.
The reduction algorithm outputs
\begin{align*}
    \mathbf{A}' &:=
    \begin{pmatrix}
        \mathbf{A} & \mathbf{O}_{m \times k} \\
        C\mathbf{T}\mathbf{A} & -C\mathbf{I}_k
    \end{pmatrix}
    \in \mathbb{Z}^{(m+k)\times(n+k)}, \\
    \vec{s}' &:=
    \begin{pmatrix}
        -\vec{s} \\
        -C\mathbf{T}\vec{s}
    \end{pmatrix}
    \in \mathbb{Z}^{m+k},
    \qquad
    \ell' := \ell, \\
    \mathbf{T}' &:=
    \begin{pmatrix}
        \mathbf{O}_{k\times n} & \mathbf{I}_k
    \end{pmatrix}
    \in \mathbb{Z}^{k\times(n+k)},
\end{align*}
where $C := 1 + \lceil \max\{1,\alpha\}\ell \rceil$, $\mathbf{O}_{m\times k}$ is the $m\times k$ zero matrix, $\mathbf{O}_{k\times n}$ is the $k\times n$ zero matrix, and $\mathbf{I}_k$ is the $k\times k$ identity matrix.

It is clear that the reduction runs in polynomial time.

Moreover, $\mathbf{A}'$ has column rank $n+k$.
Indeed, suppose that $\mathbf{A}' (\vec{u}, \vec{w})^T = \vec{0}$ for some $\vec{u} \in \mathbb{Z}^n$ and $\vec{w} \in \mathbb{Z}^k$.
Then $\mathbf{A}\vec{u} = \vec{0}$.
Since $\mathbf{A}$ is a lattice basis, it has column rank $n$, and hence $\vec{u} = \vec{0}$.
Substituting this into the lower block gives $-C\vec{w} = \vec{0}$, and therefore $\vec{w} = \vec{0}$.
Thus, $\mathbf{A}'$ has column rank $n+k$.

Next, $(\mathbf{A}',\vec{s}',\ell',\mathbf{T}')$ satisfies the promise of \Cref{def:var-ldlp}, namely, $\lambda_1^{(p)}(\mathcal{L}(\mathbf{A}')) \geq \ell'$.
Indeed, for every nonzero vector $(\vec{u}, \vec{w})^T \in \mathbb{Z}^{n+k}$, we have
\begin{align*}
    \mathbf{A}'
    \begin{pmatrix}
        \vec{u} \\
        \vec{w}
    \end{pmatrix}
    =
    \begin{pmatrix}
        \mathbf{A}\vec{u} \\
        C(\mathbf{T}\mathbf{A}\vec{u}-\vec{w})
    \end{pmatrix}.
\end{align*}
If $\vec{u} \neq \vec{0}$, then
\begin{align*}
    \left\|
    \mathbf{A}'
    \begin{pmatrix}
        \vec{u} \\
        \vec{w}
    \end{pmatrix}
    \right\|_p
    \geq
    \|\mathbf{A}\vec{u}\|_p
    \geq
    \lambda_1^{(p)}(\mathcal{L}(\mathbf{A}))
    \geq
    \ell
    =
    \ell'.
\end{align*}
If $\vec{u} = \vec{0}$, then $\vec{w} \neq \vec{0}$, and hence
\begin{align*}
    \left\|
    \mathbf{A}'
    \begin{pmatrix}
        \vec{0} \\
        \vec{w}
    \end{pmatrix}
    \right\|_p
    =
    \left\|
    \begin{pmatrix}
        \vec{0} \\
        -C\vec{w}
    \end{pmatrix}
    \right\|_p
    =
    C\|\vec{w}\|_p
    \geq
    C
    >
    \ell
    =
    \ell'.
\end{align*}
Therefore, $\lambda_1^{(p)}(\mathcal{L}(\mathbf{A}')) \geq \ell'$.

We first show the completeness.
Assume that $(\mathbf{A}, \vec{s}, \ell, \mathbf{T})$ is a YES instance of $(p,\alpha)$-$\CosetLDLP$.
Fix an arbitrary vector $\vec{y} \in \{0,1\}^k$.
Then there exists a vector $\vec{v} \in (\vec{s} + \mathcal{L}(\mathbf{A})) \cap \mathcal{B}_p^m(\alpha\ell)$ such that $\mathbf{T}\vec{v} = \vec{y}$.
Since $\vec{v} \in \vec{s} + \mathcal{L}(\mathbf{A})$, there exists $\vec{z} \in \mathbb{Z}^n$ such that $\vec{v} = \vec{s} + \mathbf{A}\vec{z}$.
Define $\vec{x} := (\vec{z}, \vec{y})^T \in \mathbb{Z}^{n+k}$.
Then $\mathbf{T}'\vec{x} = \vec{y}$.
Moreover,
\begin{align*}
    \mathbf{A}'\vec{x}-\vec{s}'
    &=
    \begin{pmatrix}
        \mathbf{A}\vec{z} \\
        C\mathbf{T}\mathbf{A}\vec{z} - C\vec{y}
    \end{pmatrix}
    -
    \begin{pmatrix}
        -\vec{s} \\
        -C\mathbf{T}\vec{s}
    \end{pmatrix} \\
    &=
    \begin{pmatrix}
        \mathbf{A}\vec{z}+\vec{s} \\
        C\mathbf{T}(\mathbf{A}\vec{z}+\vec{s}) - C\vec{y}
    \end{pmatrix} \\
    &=
    \begin{pmatrix}
        \vec{v} \\
        C(\mathbf{T}\vec{v}-\vec{y})
    \end{pmatrix}
    =
    \begin{pmatrix}
        \vec{v} \\
        \vec{0}
    \end{pmatrix}.
\end{align*}
Hence,
\begin{align*}
    \|\mathbf{A}'\vec{x}-\vec{s}'\|_p
    =
    \left\|
    \begin{pmatrix}
        \vec{v} \\
        \vec{0}
    \end{pmatrix}
    \right\|_p
    =
    \|\vec{v}\|_p
    \leq
    \alpha\ell
    =
    \alpha\ell'.
\end{align*}
Since $\vec{y}$ was arbitrary, this shows that $(\mathbf{A}',\vec{s}',\ell',\mathbf{T}')$ is a YES instance of $(p,\alpha,0)$-$\CoeffLDLP$.

Next, we show the soundness.
Assume that $(\mathbf{A}, \vec{s}, \ell, \mathbf{T})$ is a NO instance of $(p,\alpha)$-$\CosetLDLP$.
Then there exists a vector $\vec{y} \in \{0,1\}^k$ such that for all $\vec{v} \in \vec{s} + \mathcal{L}(\mathbf{A})$, $\mathbf{T}\vec{v} = \vec{y}$ implies $\|\vec{v}\|_p > \alpha\ell$.

Suppose, for contradiction, that $(\mathbf{A}',\vec{s}',\ell',\mathbf{T}')$ is a YES instance of $(p,\alpha,0)$-$\CoeffLDLP$ for this vector $\vec{y}$.
Then there exists $\vec{x} = (\vec{z}, \vec{w})^T \in \mathbb{Z}^{n+k}$ such that $\mathbf{T}'\vec{x} = \vec{y}$ and $\|\mathbf{A}'\vec{x}-\vec{s}'\|_p \leq \alpha\ell'$.
Since $\mathbf{T}'\vec{x} = \vec{w}$, it follows that $\vec{w} = \vec{y}$.
Therefore,
\begin{align*}
    \mathbf{A}'\vec{x}-\vec{s}'
    &=
    \begin{pmatrix}
        \mathbf{A}\vec{z}+\vec{s} \\
        C\mathbf{T}(\mathbf{A}\vec{z}+\vec{s}) - C\vec{y}
    \end{pmatrix}.
\end{align*}
If $\mathbf{T}(\mathbf{A}\vec{z}+\vec{s}) - \vec{y} \neq \vec{0}$, then the lower block is a nonzero integer vector, and hence
\begin{align*}
    \|\mathbf{A}'\vec{x}-\vec{s}'\|_p
    \geq
    \left\|C\bigl(\mathbf{T}(\mathbf{A}\vec{z}+\vec{s})-\vec{y}\bigr)\right\|_p
    \geq
    C
    >
    \alpha\ell
    =
    \alpha\ell',
\end{align*}
which is a contradiction.
Thus, $\mathbf{T}(\mathbf{A}\vec{z}+\vec{s}) = \vec{y}$.
Set $\vec{v} := \mathbf{A}\vec{z}+\vec{s}$.
Then $\vec{v} \in \vec{s} + \mathcal{L}(\mathbf{A})$ and $\mathbf{T}\vec{v} = \vec{y}$.
Moreover,
\begin{align*}
    \|\vec{v}\|_p
    \leq
    \|\mathbf{A}'\vec{x}-\vec{s}'\|_p
    \leq
    \alpha\ell'
    =
    \alpha\ell,
\end{align*}
which contradicts the choice of $\vec{y}$.
Therefore, $(\mathbf{A}',\vec{s}',\ell',\mathbf{T}')$ is a NO instance of $(p,\alpha,0)$-$\CoeffLDLP$.
\end{proof}

From the above, the following theorem follows immediately.
\begin{theorem}
\label{thm:ldlp-variant-reductions}
For any $p \in [1, \infty]$, $\alpha > 0$, $\beta > 0$ and $k \geq 1$, the following holds:
\begin{itemize}
    \item There exists a deterministic polynomial-time many-one reduction from $(p, \alpha, \beta)$-$\CoeffLDLP$ to $(p, \alpha + \beta)$-$\CosetLDLP$.
    \item There exists a deterministic polynomial-time many-one reduction from $(p, \alpha)$-$\CosetLDLP$ to $(p, \alpha, 0)$-$\CoeffLDLP$.
\end{itemize}
Moreover, in both reductions, the number \(k\) of hypercube dimension is preserved.
\end{theorem}

By considering the reduction of \Cref{prop:coefficient-to-coset} and \Cref{prop:coset-to-coefficient} as algorithms with YES instances as input, \Cref{thm:among-ldls} is immediately obtained.

\section{Hardness of LDLP}
\label{sec:hardnes}
In this section, we show that $(p, \alpha)$-$\LDLP$ is $\Pi_2^p$-complete via a reduction from $\forall\exists$ 1-in-3-SAT.
$\forall\exists$ 1-in-3-SAT is formally defined as follows:
An instance consists of a Boolean formula $\psi(\vec{u},\vec{v})$ in CNF, where $\vec{u} = (u_1,\dots,u_n)$ are universally quantified variables and $\vec{v} = (v_1,\dots,v_m)$ are existentially quantified variables, and every clause contains exactly three literals.
A truth assignment satisfies a clause if and only if exactly one of its three literals evaluates to $1$, where a literal is either a variable $x$ or its negation $\neg x$.
The quantified formula is
\begin{align*}
    \Phi \;=\; \forall \vec{u}\in\{0,1\}^n\ \exists \vec{v}\in\{0,1\}^{m}:\ \psi(\vec{u},\vec{v}),
\end{align*}
and the decision problem asks whether $\Phi$ is true.

\begin{theorem}[\cite{BJORKLUND2011450}, adapted]
\label{thm:sat}
$\forall\exists$ 1-in-3-SAT is $\Pi_2^p$-complete.
\end{theorem}

To reduce from $\forall \exists$ 1-in-3-SAT to LDLP, we construct the following gadget.
The following elementary encoding of clauses is similar to an encoding that appears in the reduction of Bennett, Golovnev, and Stephens-Davidowitz \cite{conf/focs/BennettGS17} to Closest Vector Problem (CVP).
\begin{lemma}
\label{lem:1in3-clause-gadget}
Let $\psi$ be a Boolean formula on variables $x_1,\ldots,x_n$, where every clause
contains exactly three literals.
Let $m_\psi$ be the number of clauses of $\psi$.
Then there exists a deterministic polynomial-time algorithm that outputs
\[
    \mathbf{C}\in\mathbb Z^{m_\psi\times n},
    \qquad
    \vec b\in\mathbb Z^{m_\psi},
\]
such that for every $\vec z\in\{0,1\}^n$,
\begin{align}
\label{eq:1in3-gadget-sat}
    \mathbf{C}\vec z=\vec b
    &\iff
    \vec z \text{ satisfies } \psi \text{ in the 1-in-3 sense},\\
\label{eq:1in3-gadget-gap}
    \mathbf{C}\vec z\neq \vec b
    &\Longrightarrow
    \norm{\mathbf{C}\vec z-\vec b}_\infty\ge 2.
\end{align}
\end{lemma}

\begin{proof}
For each clause $C_j$ of $\psi$, write $C_j=(\tau_{j,1},\tau_{j,2},\tau_{j,3})$, where each $\tau_{j,r}$ is either $x_h$ or $\neg x_h$ for some $h\in[n]$.
For each $(j,r)$, let $\iota(j,r)\in[n]$ be the unique index such that $\tau_{j,r}\in\{x_{\iota(j,r)},\neg x_{\iota(j,r)}\}$.
Let
\[
    P_j:=\{r\in\{1,2,3\}:\tau_{j,r}=x_{\iota(j,r)}\},
    \qquad
    N_j:=\{r\in\{1,2,3\}:\tau_{j,r}=\neg x_{\iota(j,r)}\},
\]
and set $c_j:=|N_j|$.

Define $\mathbf{C}\in\mathbb Z^{m_\psi\times n}$ and $\vec b\in\mathbb Z^{m_\psi}$ by
\[
     C_{j,h}
    :=
    2\cdot \bigl|\{r\in P_j:\iota(j,r)=h\}\bigr|
    -
    2\cdot \bigl|\{r\in N_j:\iota(j,r)=h\}\bigr|
    \qquad (1\le h\le n),
\]
and $b_j:=2(1-c_j)$.
Fix any $\vec z\in\{0,1\}^n$.
For each $j\in[m_\psi]$, define
\[
    w_j(\vec z)
    :=
    \sum_{r\in P_j} z_{\iota(j,r)}
    +
    \sum_{r\in N_j} (1-z_{\iota(j,r)}).
\]
This is exactly the number of literals of $C_j$ that evaluate to true under $\vec z$.
By the definition of $\mathbf{C}$,
\begin{align*}
    (\mathbf{C}\vec z)_j
    &=
    2\sum_{r\in P_j} z_{\iota(j,r)}
    -
    2\sum_{r\in N_j} z_{\iota(j,r)}\\
    &=
    2\sum_{r\in P_j} z_{\iota(j,r)}
    +
    2\sum_{r\in N_j} (1-z_{\iota(j,r)})
    -
    2|N_j|\\
    &=
    2w_j(\vec z)-2c_j.
\end{align*}
Since $b_j=2(1-c_j)$, we obtain
\[
    (\mathbf{C}\vec z-\vec b)_j
    =
    (2w_j(\vec z)-2c_j)-2(1-c_j)
    =
    2\bigl(w_j(\vec z)-1\bigr).
\]
Hence
\[
    (\mathbf{C}\vec z-\vec b)_j=0
    \iff
    w_j(\vec z)=1
    \iff
    C_j \text{ is satisfied in the 1-in-3 sense under }\vec z.
\]
Since this holds for every $j\in[m_\psi]$, we get \Cref{eq:1in3-gadget-sat}.

If $\mathbf{C}\vec z\neq \vec b$, then by \Cref{eq:1in3-gadget-sat} the
assignment $\vec z$ does not satisfy $\psi$ in the 1-in-3 sense.
Hence there exists some clause index $j$ such that $w_j(\vec z)\in\{0,2,3\}$.
For this $j$,
\[
    \abs{(\mathbf{C}\vec z-\vec b)_j}
    =
    2\abs{w_j(\vec z)-1}
    \ge 2.
\]
Therefore $\norm{\mathbf{C}\vec z-\vec b}_\infty\ge 2$,
which proves \Cref{eq:1in3-gadget-gap}.
The construction is deterministic and polynomial-time.
\end{proof}

\subsection{For the Finite Norm}
We first prove that, for every finite $p \geq \log_2 3$, $(p, \alpha)$-$\LDLP$ in the $\ell_p$ norm is complete for the second level of the polynomial hierarchy.

\begin{proof}[Proof of Item 1 of \Cref{thm:ldlp-is-pi2-complete-lp}]
By \Cref{prop:pi2p}, it suffices to prove $\Pi_2^p$-hardness. We reduce from $\forall\exists$~1-in-3-SAT, which is $\Pi_2^p$-complete by \Cref{thm:sat}.

Let
\[
    \Phi=\forall \vec{u}\in\{0,1\}^{k}\ \exists \vec{v}\in\{0,1\}^{t}:\ \psi(\vec{u},\vec{v})
\]
be an instance of $\forall\exists$~1-in-3-SAT, where $\psi$ is a Boolean formula in which every clause contains exactly three literals. Let $m_\psi$ be the number of clauses of $\psi$, and set $n:=k+t$.

We first describe the output of the reduction.
Apply \Cref{lem:binary-rs-gadget-lp} with input $r=n$. We obtain a lattice basis $\mathbf{A}\in\mathbb Z^{N\times d}$, an integer vector $\vec{s}\in\mathbb Z^N$, a positive integer $\ell$, and the projection matrix $\mathbf{P}_n\in\mathbb Z^{n\times N}$ onto the first $n$ coordinates such that
\begin{align}
\label{eq:rs-cover}
    \{0,1\}^n
    &\subseteq
    \{\mathbf{P}_n\vec{x}:\vec{x}\in (\vec{s}+\mathcal{L}(\mathbf{A}))\cap \mathcal{B}_p^N(\alpha\ell)\},\\
\label{eq:rs-binary}
    (\vec{s}+\mathcal{L}(\mathbf{A}))\cap \mathcal{B}_p^N(\alpha\ell)
    &\subseteq
    \{0,1\}^N,
\end{align}
and $\lambda_1^{(p)}(\mathcal{L}(\mathbf{A}))\ge \ell$.
Apply \Cref{lem:1in3-clause-gadget} to $\psi$. We obtain $\widetilde{\mathbf C}\in\mathbb Z^{m_\psi\times n}$ and $\vec b\in\mathbb Z^{m_\psi}$ such that for every $\vec z\in\{0,1\}^n$,
\begin{align}
\label{eq:clause-gadget-sat}
    \widetilde{\mathbf C}\vec z=\vec b
    &\iff
    \vec z \text{ satisfies } \psi \text{ in the 1-in-3 sense},\\
\label{eq:clause-gadget-gap}
    \widetilde{\mathbf C}\vec z\neq \vec b
    &\Longrightarrow
    \norm{\widetilde{\mathbf C}\vec z-\vec b}_\infty\ge 2.
\end{align}
Define $\mathbf{C}:=\widetilde{\mathbf C}\mathbf{P}_n\in\mathbb Z^{m_\psi\times N}$, let $\mathbf{P}_k\in\mathbb Z^{k\times N}$ be the projection onto the first $k$ coordinates, and set $M:=\ell+1$. We output the tuple $(\mathbf{A}',\vec{s}',\ell',\mathbf{T}')$ defined by
\[
    \mathbf{A}':=
    \begin{pmatrix}
        M\mathbf{C}\mathbf{A}\\
        \mathbf{A}
    \end{pmatrix}
    \in\mathbb Z^{(m_\psi+N)\times d},
    \qquad
    \vec{s}':=
    \begin{pmatrix}
        M(\mathbf{C}\vec{s}-\vec b)\\
        \vec{s}
    \end{pmatrix}
    \in\mathbb Z^{m_\psi+N},
\]
\[
    \ell':=\ell,
    \qquad
    \mathbf{T}':=
    \begin{pmatrix}
        \mathbf{O}_{k\times m_\psi} & \mathbf{P}_k
    \end{pmatrix}
    \in\mathbb Z^{k\times (m_\psi+N)}.
\]
This map is computable deterministically in polynomial time.

We verify Item~\ref{item:ldlp-min-dist} of \Cref{def:ldlp}. Since $\mathbf{A}$ has column rank $d$, the matrix $\mathbf{A}'$ also has column rank $d$. For every nonzero $\vec z\in\mathbb Z^d\setminus\{\vec 0\}$,
\[
    \norm{\mathbf{A}'\vec z}_p^p
    =
    M^p\norm{\mathbf{C}\mathbf{A}\vec z}_p^p+\norm{\mathbf{A}\vec z}_p^p
    \ge
    \norm{\mathbf{A}\vec z}_p^p
    \ge
    \ell^p
    =
    (\ell')^p.
\]
Hence $\lambda_1^{(p)}(\mathcal{L}(\mathbf{A}'))\ge \ell'$.

We next prove completeness. Assume that $\Phi$ is true. Fix an arbitrary vector $\vec u\in\{0,1\}^k$. Since $\Phi$ is true, there exists $\vec v\in\{0,1\}^t$ such that the assignment $\vec a:=(\vec u,\vec v)\in\{0,1\}^n$ satisfies $\psi$ in the 1-in-3 sense. By \Cref{eq:rs-cover}, there exists $\vec x\in (\vec{s}+\mathcal{L}(\mathbf{A}))\cap \mathcal{B}_p^N(\alpha\ell)$ such that $\mathbf{P}_n\vec x=\vec a$. By \Cref{eq:rs-binary}, we have $\vec x\in\{0,1\}^N$. Therefore
\[
    \mathbf{C}\vec x
    =
    \widetilde{\mathbf C}\mathbf P_n\vec x
    =
    \widetilde{\mathbf C}\vec a
    =
    \vec b,
\]
where the last equality follows from \Cref{eq:clause-gadget-sat}.
Choose $\vec z\in\mathbb Z^d$ with $\vec x=\vec{s}+\mathbf{A}\vec z$, and define
\[
    \vec w:=\vec s'+\mathbf A'\vec z
    =
    \begin{pmatrix}
        M(\mathbf C\vec x-\vec b)\\
        \vec x
    \end{pmatrix}
    =
    \begin{pmatrix}
        \vec 0\\
        \vec x
    \end{pmatrix}.
\]
Then $\vec w\in \vec s'+\mathcal L(\mathbf A')$ and $\norm{\vec w}_p=\norm{\vec x}_p\le \alpha\ell=\alpha\ell'$. Moreover, $\mathbf T'\vec w=\mathbf P_k\vec x=\vec u$. Since $\vec u\in\{0,1\}^k$ was arbitrary, we obtain
\[
    \{0,1\}^k
    \subseteq
    \{\mathbf T'\vec w:\vec w\in (\vec s'+\mathcal L(\mathbf A'))\cap \mathcal B_p^{m_\psi+N}(\alpha\ell')\}.
\]
Therefore, $(\mathbf A',\vec s',\ell',\mathbf T')$ is a YES instance of $(p, \alpha)$-$\LDLP$.

Finally, we prove soundness. Assume that $\Phi$ is false. Then there exists $\vec u^\star\in\{0,1\}^k$ such that for every $\vec v\in\{0,1\}^t$, the assignment $(\vec u^\star,\vec v)$ does not satisfy $\psi$ in the 1-in-3 sense.
We claim that
\[
    \vec u^\star\notin
    \{\mathbf T'\vec w:\vec w\in (\vec s'+\mathcal L(\mathbf A'))\cap \mathcal B_p^{m_\psi+N}(\alpha\ell')\}.
\]
Suppose toward a contradiction that there exists $\vec w\in (\vec s'+\mathcal L(\mathbf A'))\cap \mathcal B_p^{m_\psi+N}(\alpha\ell')$ such that $\mathbf T'\vec w=\vec u^\star$. Since $\vec w\in \vec s'+\mathcal L(\mathbf A')$, there exists $\vec z\in\mathbb Z^d$ such that
\[
    \vec w=\vec s'+\mathbf A'\vec z
    =
    \begin{pmatrix}
        M(\mathbf C\vec x-\vec b)\\
        \vec x
    \end{pmatrix},
    \qquad
    \vec x:=\vec{s}+\mathbf{A}\vec z.
\]
Because the lower block of $\vec w$ is $\vec x$, we have $\norm{\vec x}_p\le \norm{\vec w}_p\le \alpha\ell'=\alpha\ell$. Hence \Cref{eq:rs-binary} implies that $\vec x\in\{0,1\}^N$. Let $\vec a:=\mathbf P_n\vec x\in\{0,1\}^n$. Since $\mathbf P_k\vec x=\mathbf T'\vec w=\vec u^\star$, the first $k$ coordinates of $\vec a$ are exactly $\vec u^\star$. Therefore $\vec a$ does not satisfy $\psi$ in the 1-in-3 sense. By \Cref{eq:clause-gadget-sat}, we have $\widetilde{\mathbf C}\vec a\neq \vec b$. Hence \Cref{eq:clause-gadget-gap} gives $\norm{\widetilde{\mathbf C}\vec a-\vec b}_\infty\ge 2$. Since
\[
    \mathbf C\vec x-\vec b
    =
    \widetilde{\mathbf C}\mathbf P_n\vec x-\vec b
    =
    \widetilde{\mathbf C}\vec a-\vec b,
\]
we obtain $\norm{\mathbf C\vec x-\vec b}_\infty\ge 2$. Therefore there exists some clause index $j$ such that $\abs{(\mathbf C\vec x-\vec b)_j}\ge 2$. Since the $j$-th coordinate of the upper block of $\vec w$ equals $M(\mathbf C\vec x-\vec b)_j$, it follows that
\[
    \norm{\vec w}_p
    \ge
    M\abs{(\mathbf C\vec x-\vec b)_j}
    \ge
    2M
    >
    \ell
    \ge
    \alpha\ell
    =
    \alpha\ell',
\]
which contradicts $\norm{\vec w}_p\le \alpha\ell'$.
Thus no such $\vec w$ exists, and $(\mathbf A',\vec s',\ell',\mathbf T')$ is a NO instance of $(p, \alpha)$-$\LDLP$.

Hence the reduction is correct, and $(p, \alpha)$-$\LDLP$ is $\Pi_2^p$-hard. Together with \Cref{prop:pi2p}, this proves the theorem.
\end{proof}

We show that, for every $p \geq 1$, $\LDLP$ is $\NP$-hard and $\coNP$-hard.
We define a variant of the Closest Vector Problem (CVP), which is called $\GapCVP'$. This variant includes a promise that, in YES instances, there exist lattice vectors close to the target vector whose coefficient vectors are binary.  In NO instances, it is promised that all lattice vectors are far from any non-zero integer multiple of the target vector.
\begin{definition}[$\gamma$-$\GapCVP'$, \cite{journals/jcss/AroraBSS97}]
    \label{def:cvp}
    For $p \geq 1$ and $\gamma=\gamma(n)\geq 1$, the promise problem $\gamma$-approximate Closest Vector Problem ($\gamma$-$\GapCVP'_p$) is defined as follows. Instances are triples $(\mathbf{B}, \vec{t}, d)$, where $\mathbf{B} \in \mathbb{Z}^{m\times n}$ is a lattice basis of $\mathcal{L}(\mathbf{B})$, $\vec{t}\in \mathbb{Z}^m $ is a target vector, and $d>0$ is a distance threshold such that
    \begin{enumerate}
        \item $(\mathbf{B}, \vec{t}, d)$ is a YES instance if $\lVert\mathbf{B}\vec{x} - \vec{t}\rVert_p\leq d$ for some $\vec{x}\in \{0,1\}^n$;
        \item$(\mathbf{B}, \vec{t}, d)$ is a NO instance if dist$_p(\mathcal{L}(\mathbf{B}), w\vec{t}) > \gamma\cdot d$ for all $w\in\mathbb{Z}\backslash\{0\}$.
    \end{enumerate}
\end{definition}

The following theorem provides the \NP-hardness of $\gamma$-$\GapCVP'_p$.
\begin{theorem}[\cite{journals/jcss/AroraBSS97}]
For any $p\in [1,\infty)$ and constant $\gamma \geq 1$, $\gamma$-$\GapCVP'_p$ is \NP-hard.
\end{theorem}

The following theorem is proved via a reduction from $\gamma$-$\GapCVP'_p$ to $(p, \alpha, \beta)$-$\CoeffLDLP$. This proof is very similar to the proof of the \NP-hardness of $\alpha$-Bounded Distance Decoding ($\alpha$-BDD) by Liu, Lyubashevsky and Micciancio \cite{Liu2006OnBD}.

\begin{theorem}
\label{lem:coff-ldlp-is-np-hard}
    For any $p \in [1, \infty)$, $\alpha > 1/2^{1/p}$ and $\beta \ge 0$, $(p, \alpha, \beta)$-$\CoeffLDLP$ is $\NP$-hard under deterministic polynomial-time many-one reductions.
\end{theorem}
\begin{proof}
Since $\alpha > 2^{-1/p}$, fix a constant
$\sigma \in (2^{-1/p}, \min\{\alpha,1\})$, and choose a constant
$\gamma \ge 1$ such that
\[
    \gamma > \frac{\alpha+\beta}{(\alpha^p-\sigma^p)^{1/p}} .
\]
We reduce from $\gamma$-$\GapCVP'_p$.
Let $(\mathbf{B}, \vec{t}, d)$ be a $\gamma$-$\GapCVP'_p$ instance of rank $n$.
After clearing denominators in $d$, if necessary, we may assume that
$d \in \mathbb{Z}_{>0}$; this preserves the YES/NO status of the instance and changes the
bit length only polynomially. Let $L$ denote the bit length of the resulting instance.

Choose and hardwire rational constants $\theta$ and $\kappa$ such that
\[
    \frac{1}{(\alpha^p-\sigma^p)^{1/p}}
    < \theta
    < \frac{\gamma}{\alpha+\beta},
    \qquad
    \kappa > \alpha+1 .
\]
Such a rational $\theta$ exists by the choice of $\gamma$.
The reduction algorithm computes a $(p,\sigma,k)$-locally dense lattice
$(\mathbf{A}, \vec{s}, \ell, \mathbf{T})$ with $k := (n+L)^2$ using the algorithm from
\Cref{lem:binary-rs-gadget-lp}. Let $r$ denote the rank of $\mathbf{A}$.
Define
\[
    \zeta := \frac{d}{\ell}\theta,
    \qquad
    \Delta := \lceil \kappa \zeta \ell \rceil
             = \lceil \kappa \theta d \rceil .
\]
We first define the following rational tuple:
\begin{align*}
    \widehat{\mathbf{A}} &:= \begin{pmatrix}
        \mathbf{B}\mathbf{P}\mathbf{T}\mathbf{A} &  \mathbf{B}\mathbf{P}\mathbf{T}\vec{s} \\
        \zeta \mathbf{A} & \zeta \vec{s} \\
        0 & \Delta
    \end{pmatrix},
    \\
   \widehat{\vec{s}} &:=
    \begin{pmatrix}
        \vec{t} \\
        \vec{0} \\
        \Delta
    \end{pmatrix},
    \qquad
    \widehat{\ell} := \zeta \ell,
    \\
    \widehat{\mathbf{T}} &:=
    \begin{pmatrix}
        \mathbf{Q}\mathbf{T}\mathbf{A}
        & \mathbf{Q}\mathbf{T}\vec{s}
    \end{pmatrix},
\end{align*}
where $\mathbf{I}_{k'}$ is the $k'\times k'$ identity matrix,
$\mathbf{P} := (\mathbf{I}_n \mid \mathbf{O}_{n\times k'}) \in \mathbb{Z}^{n\times k}$
is the projection onto the first $n$ coordinates,
$\mathbf{Q} := (\mathbf{O}_{k'\times n}\mid \mathbf{I}_{k'}) \in \mathbb{Z}^{k'\times k}$
is the projection onto the last $k'$ coordinates, and $k' := k-n$.
Then, from the definitions of $\theta$ and $\zeta$, we obtain the inequality
\begin{align}
    \label{eq:1}
    \frac{d}{\ell(\alpha^p-\sigma^p)^{1/p}} < \zeta < \frac{\gamma d}{(\alpha+\beta)\ell}.
\end{align}

Since $\theta$, $d$, and $\ell$ are rational or integral as above, $\zeta$ is rational.
Choose a positive integer $D$ such that $D\zeta \in \mathbb{Z}$ and
$D\zeta\ell \in \mathbb{Z}$. The reduction outputs
\[
    \mathbf{A}' := D\widehat{\mathbf{A}},
    \qquad
    \vec{s}' := D\widehat{\vec{s}},
    \qquad
    \ell' := D\widehat{\ell},
    \qquad
    \mathbf{T}' := \widehat{\mathbf{T}} .
\]
Then all entries of $\mathbf{A}'$ and $\vec{s}'$ are integral. Moreover, since
$\theta$ and $\kappa$ are fixed constants and the tuple
$(\mathbf{A},\vec{s},\ell,\mathbf{T})$ has polynomial bit length, the integers
$\Delta$ and $D$ also have polynomial bit length. Hence the construction runs in
deterministic polynomial time.

It remains to verify correctness. We first record that scaling by $D$ preserves all relevant
conditions. For every $\vec{u}\in\mathbb{Z}^{r+1}$,
\[
    \|\mathbf{A}'\vec{u}-\vec{s}'\|_p
    =
    D\|\widehat{\mathbf{A}}\vec{u}-\widehat{\vec{s}}\|_p,
    \qquad
    \mathbf{T}'\vec{u}=\widehat{\mathbf{T}}\vec{u},
\]
and
\[
    \lambda_1^{(p)}(\mathcal{L}(\mathbf{A}'))
    =
    D\lambda_1^{(p)}(\mathcal{L}(\widehat{\mathbf{A}})).
\]
Therefore it suffices to prove the promise, completeness, and soundness for the rational
tuple $(\widehat{\mathbf{A}},\widehat{\vec{s}},\widehat{\ell},\widehat{\mathbf{T}})$.

First, we show that
$(\widehat{\mathbf{A}},\widehat{\vec{s}},\widehat{\ell},\widehat{\mathbf{T}})$ satisfies the
promise of $(p,\alpha,\beta)$-$\CoeffLDLP$.
The matrix $\widehat{\mathbf{A}}$ has column rank $r+1$: if
$\widehat{\mathbf{A}}(\vec{z},w)^T=0$, then the last block gives $\Delta w=0$, hence
$w=0$, and then the middle block gives $\zeta\mathbf{A}\vec{z}=0$, so $\vec{z}=0$ because
$\mathbf{A}$ has column rank $r$.

We will show
$\lambda_1^{(p)}(\mathcal{L}(\widehat{\mathbf{A}})) \ge \widehat{\ell}$.
Consider any non-zero integer vector $(\vec{z}, w) \neq (\vec{0}, 0)$. If $w=0$, then since
$\vec{z} \neq \vec{0}$, we have
\begin{align*}
    \lVert \widehat{\mathbf{A}}(\vec{z}, 0)^T \rVert_p^p
    \geq
    \zeta^p \lVert \mathbf{A}\vec{z} \rVert_p^p
    \ge
    (\zeta \ell)^p
    =
    \widehat{\ell}^{\,p}.
\end{align*}
Otherwise,
\begin{align*}
    \lVert \widehat{\mathbf{A}}(\vec{z}, w)^T \rVert_p^p
    \geq
    |\Delta w|^p
    \ge
    \Delta^p
    \ge
    (\zeta\ell)^p
    =
    \widehat{\ell}^{\,p},
\end{align*}
where we used $\Delta=\lceil \kappa\zeta\ell\rceil$ and $\kappa>1$.

Next, we show completeness.
Assume that $(\mathbf{B}, \vec{t}, d)$ is a YES instance of $\gamma$-$\GapCVP'_p$.
Then, there exists a vector $\vec{x} \in \{0,1\}^n$ such that
$\lVert \mathbf{B}\vec{x} - \vec{t} \rVert_p \le d$.
By \Cref{item:conv-hypercube} of \Cref{def:ldl}, for any binary vector
$\vec{y} \in \{0,1\}^{k'}$, there exists a vector
$\vec{v} = \vec{s} + \mathbf{A}\vec{z} \in \vec{s} + \mathcal{L}(\mathbf{A})$
for some $\vec{z} \in \mathbb{Z}^r$ such that
$\lVert \vec{v} \rVert_p^p \le \sigma^p \ell^p$ and
$\mathbf{T}\vec{v} = (\vec{x}, \vec{y})^T$.
For such a vector $\vec{z}$, we have
\begin{align*}
\begin{split}
    \lVert \widehat{\mathbf{A}}(\vec{z},1)^T - \widehat{\vec{s}} \rVert_p^p
    &= \lVert \mathbf{B}\mathbf{P}\mathbf{T}\mathbf{A}\vec{z}
        + \mathbf{B}\mathbf{P}\mathbf{T}\vec{s}
        - \vec{t} \rVert_p^p
        + \zeta^p \lVert \mathbf{A}\vec{z} + \vec{s} \rVert_p^p \\
    &= \lVert \mathbf{B}\mathbf{P}\mathbf{T}(\mathbf{A}\vec{z} + \vec{s})
        - \vec{t} \rVert_p^p
        + \zeta^p \lVert \mathbf{A}\vec{z} + \vec{s} \rVert_p^p \\
    &= \lVert \mathbf{B}\mathbf{P}(\vec{x},\vec{y})^T - \vec{t} \rVert_p^p
        + \zeta^p \lVert \vec{v} \rVert_p^p \\
    &= \lVert \mathbf{B}\vec{x} - \vec{t} \rVert_p^p
        + \zeta^p \lVert \vec{v} \rVert_p^p \\
    &\le d^p + \sigma^p \zeta^p \ell^p.
\end{split}
\end{align*}
Using the first inequality in \Cref{eq:1}, we obtain
\begin{align*}
    \lVert \widehat{\mathbf{A}}(\vec{z},1)^T - \widehat{\vec{s}} \rVert_p^p
    \le
    d^p + \sigma^p \zeta^p \ell^p
    <
    \alpha^p \zeta^p \ell^p
    =
    \alpha^p \widehat{\ell}^{\,p}.
\end{align*}
Moreover,
\begin{align*}
    \widehat{\mathbf{T}}(\vec{z},1)^T
    &= \mathbf{Q}\mathbf{T}(\mathbf{A}\vec{z} + \vec{s}) \\
    &= \mathbf{Q}\mathbf{T}\vec{v} \\
    &= \vec{y} \in \{0,1\}^{k'}.
\end{align*}
This implies that
\begin{align*}
    \{0,1\}^{k'} \subseteq
    \{\widehat{\mathbf{T}}\vec{u} :
        \vec{u} \in \mathbb{Z}^{r+1},\,
        \lVert \widehat{\mathbf{A}}\vec{u} - \widehat{\vec{s}} \rVert_p
        \le \alpha \widehat{\ell} \}.
\end{align*}

Finally, we show soundness.
Assume that $(\mathbf{B},\vec{t},d)$ is a NO instance of $\gamma$-$\GapCVP'_p$.
Then, by \Cref{def:cvp}, taking the nonzero multiplier $1$, for every
$\vec{x}\in \mathbb{Z}^n$ it holds that
$\lVert\mathbf{B}\vec{x} - \vec{t}\rVert_p > \gamma d$.
For all $(\vec{z}, w)^T \in \mathbb{Z}^{r+1}$, the vector
$\mathbf{P}\mathbf{T}(\mathbf{A}\vec{z}+w\vec{s})$ belongs to $\mathbb{Z}^n$, and hence
\begin{align*}
    \lVert\widehat{\mathbf{A}}(\vec{z}, w)^T - \widehat{\vec{s}}\rVert^p_p
    \geq
    \lVert\mathbf{B}\mathbf{P}\mathbf{T}(\mathbf{A}\vec{z} + w\vec{s})
        - \vec{t}\rVert^p_p
    >
    \gamma^p d^p.
\end{align*}
Since $\zeta < \gamma d/\left((\alpha + \beta)\ell\right)$ from \Cref{eq:1}, we have
\begin{align*}
    \gamma^p d^p
    >
    (\alpha + \beta)^p \zeta^p \ell^p
    =
    (\alpha + \beta)^p\widehat{\ell}^{\,p}.
\end{align*}
Therefore, there is no integer vector
$\vec{u}:=(\vec{z}, w)^T \in \mathbb{Z}^{r+1}$ such that
$\lVert \widehat{\mathbf{A}}\vec{u} - \widehat{\vec{s}}\rVert_p
\leq (\alpha+\beta)\widehat{\ell}$.
It follows that
\begin{align*}
    \{0,1\}^{k'} \nsubseteq
    \{\widehat{\mathbf{T}}\vec{u} :
        \vec{u} \in \mathbb{Z}^{r+1},\,
        \lVert\widehat{\mathbf{A}}\vec{u} - \widehat{\vec{s}}\rVert_p
        \leq (\alpha+\beta)\widehat{\ell}\}.
\end{align*}

By the scaling observation above, the actual integer output
$(\mathbf{A}',\vec{s}',\ell',\mathbf{T}')$ satisfies the same completeness and soundness
conditions with threshold $\ell'=D\widehat{\ell}$, and also satisfies the promise
$\lambda_1^{(p)}(\mathcal{L}(\mathbf{A}'))\ge \ell'$.
This proves the correctness of the reduction.
\end{proof}

\begin{theorem}
\label{thm:ldlp-is-np-hard}
For any $p \in [1, \infty)$, $\alpha > 1/2^{1/p}$ and $\beta > 0$, $(p, \alpha + \beta)$-$\LDLP$ is $\NP$-hard under deterministic polynomial-time many-one reductions.
\end{theorem}
\begin{proof}
Therefore, we obtain \Cref{thm:ldlp-is-np-hard} by combining \Cref{prop:coefficient-to-coset} and \Cref{lem:coff-ldlp-is-np-hard}.
\end{proof}

We define the promise problem version of the Shortest Vector Problem, which is defined as follows:
\begin{definition}[$\gamma$-$\GapSVP_p$, \cite{GOLDREICH2000540}]
    For $p \geq 1$ and $\gamma=\gamma(n)\geq 1$, the promise problem $\gamma$-approximate Shortest Vector Problem ($\gamma$-GapSVP$_p$) is defined as follows. Instances are pairs $(\mathbf{B}, d)$, where $\mathbf{B} \in \mathbb{Z}^{m\times n}$ is a lattice basis of $\mathcal{L}(\mathbf{B})$, and $d > 0$ is a distance threshold such that
    \begin{enumerate}
        \item $(\mathbf{B}, d)$ is a YES instance if $\lambda^{(p)}_1(\mathcal{L}(\mathbf{B}))\leq d$;
        \item$(\mathbf{B}, d)$ is a NO instance if $\lambda^{(p)}_1(\mathcal{L}(\mathbf{B}))> \gamma\cdot d$.
    \end{enumerate}
\end{definition}

Wan \cite{wan2026} proved the deterministic $\NP$-hardness of $\GapSVP_p$.
\begin{theorem}[Theorem 7.4 of \cite{wan2026}]
For any $p\in [1,\infty)$ and constant $\gamma \in [1, 2^{1/p})$, $\gamma$-$\GapSVP_p$ is \NP-hard.
\end{theorem}

We show that $(p, \alpha)$-LDLP is $\coNP$-hard via a reduction from $\coGapSVP_p$ for all $p\geq 1$, which is obtained by swapping the YES and NO instances of $\gamma$-$\GapSVP_p$ for $\gamma = 1$. The problem $\coGapSVP_p$ is defined as follows: given an input $(\mathbf{B}, d)$, decide whether $\lambda_1^{(p)}(\mathcal{L}(\mathbf{B})) \ge d$ or $\lambda_1^{(p)}(\mathcal{L}(\mathbf{B})) < d$. Note that, by $\NP$-hardness of $\GapSVP_p$, $\coGapSVP_p$ is $\coNP$-hard.

\begin{theorem}
\label{thm:ldlp-is-conp-hard}
For any $p \in [1,\infty)$ and $\alpha > 1/2^{1/p}$, $(p, \alpha)$-$\LDLP$ is $\coNP$-hard under deterministic polynomial-time many-one reductions.
\end{theorem}
\begin{proof}
We give a deterministic polynomial-time many-one reduction from
$\coGapSVP_p$. Let $(B,d)$ be an instance of $\coGapSVP_p$, where
$B \in \mathbb{Z}^{m' \times n'}$ is a lattice basis and $d > 0$.
By multiplying $B$ and $d$ by a common positive integer, we may assume without loss of
generality that $d \in \mathbb{Z}_{>0}$.

Choose the constant
\[
c := \alpha^p + \frac12 .
\]
Since $\alpha > 2^{-1/p}$, we have $\alpha^p > 1/2$, and hence
\[
\alpha^p < c < 2\alpha^p
\qquad\text{and}\qquad
c > 1.
\]
Choose an integer $t \ge 1$ sufficiently large so that, letting
\[
r := \lceil c t^p \rceil - 1,
\qquad
\sigma := \alpha\left(\frac{t^p}{r+1}\right)^{1/p},
\]
we have
\[
2^{-1/p} < \sigma < 1
\qquad\text{and}\qquad
r \ge t^p.
\]
This is possible because $(r+1)/t^p \to c$, so
\[
\sigma \to \alpha c^{-1/p} \in (2^{-1/p},1),
\]
where the interval inclusion follows from $\alpha^p < c < 2\alpha^p$.

We run the deterministic algorithm of \Cref{lem:binary-rs-gadget-lp} with parameters
$(p,\sigma)$ on any integer input $k$, and obtain a tuple $(A,s,\ell,T)$, where
$A \in \mathbb{Z}^{m\times n}$ is a lattice basis, $s \in \mathbb{Z}^m$ is an integer vector, $\ell > 0$, and $T \in \mathbb{Z}^{k\times m}$ is an integer matrix.
Here we use only \Cref{item:ldl-1} and \Cref{item:ldl-2} of \Cref{lem:binary-rs-gadget-lp}.
Next, define
\begin{align*}
    \mathbf{A}_{\mathrm{ldl}} :=
    \begin{pmatrix}
        d\mathbf{A} \\
        \vdots \\
        d\mathbf{A} \\
        \mathbf{A}
    \end{pmatrix}
    \in \mathbb{Z}^{(r+1)m \times n},
    \qquad
    \vec{s}_{\mathrm{ldl}} :=
    \begin{pmatrix}
        d\vec{s} \\
        \vdots \\
        d\vec{s} \\
        \vec{s}
    \end{pmatrix}
    \in \mathbb{Z}^{(r+1)m},
\end{align*}
where there are $r$ copies of $d\mathbf{A}$ and $d\vec{s}$, and let
\begin{align*}
    \mathbf{T}_{\mathrm{ldl}}
    :=
    \begin{pmatrix}
        \mathbf{O}_{k \times rm} & \mathbf{T}
    \end{pmatrix}
    \in \mathbb{Z}^{k \times (r+1)m}.
\end{align*}
The reduction outputs
\begin{align*}
    \mathbf{A}' :=
    \begin{pmatrix}
        t\ell \cdot \mathbf{B} & \mathbf{O}_{m' \times n} \\
        \mathbf{O}_{(r+1)m \times n'} & \mathbf{A}_{\mathrm{ldl}}
    \end{pmatrix}
    \in \mathbb{Z}^{(m' + (r+1)m)\times (n'+n)},
\end{align*}
\begin{align*}
    \vec{s}' :=
    \begin{pmatrix}
        \vec{0}_{m'} \\
        \vec{s}_{\mathrm{ldl}}
    \end{pmatrix}
    \in \mathbb{Z}^{m' + (r+1)m},
    \qquad
    \mathbf{T}' :=
    \begin{pmatrix}
        \mathbf{O}_{k \times m'} & \mathbf{T}_{\mathrm{ldl}}
    \end{pmatrix}
    \in \mathbb{Z}^{k \times (m' + (r+1)m)},
    \qquad
    \ell' := td\ell.
\end{align*}
This reduction clearly runs in deterministic polynomial time.

We first consider the case where $(\mathbf{B}, d)$ is a YES instance of coGapSVP$_p$, that is, $\lambda_1^{(p)}(\mathcal{L}(\mathbf{B})) \ge d$.
We show that the output $(\mathbf{A}', \vec{s}', \ell', \mathbf{T}')$ belongs to $(p,\alpha)$-$\LDLP$.

We first verify \Cref{item:ldlp-min-dist} of \Cref{def:ldlp}.
Let $\vec{z} \in \mathbb{Z}^{n'+n}\setminus\{\vec{0}\}$, and write
$\vec{z} = (\vec{x},\vec{y})^T$ with
$\vec{x} \in \mathbb{Z}^{n'}$ and $\vec{y} \in \mathbb{Z}^{n}$.
Then
\begin{align*}
\mathbf{A}'\vec{z}
=
\begin{pmatrix}
    t\ell \cdot \mathbf{B}\vec{x} \\
    \mathbf{A}_{\mathrm{ldl}}\vec{y}
\end{pmatrix}.
\end{align*}
If $\vec{x} \neq \vec{0}$, then
\begin{align*}
    \|\mathbf{A}'\vec{z}\|_p
    \ge \|t\ell \cdot \mathbf{B}\vec{x}\|_p
    = t\ell \cdot \|\mathbf{B}\vec{x}\|_p
    \ge t\ell \cdot \lambda_1^{(p)}(\mathcal{L}(\mathbf{B}))
    \ge t\ell d
    = \ell'.
\end{align*}
If $\vec{x} = \vec{0}$, then $\vec{y} \neq \vec{0}$, and
\begin{align*}
    \|\mathbf{A}_{\mathrm{ldl}}\vec{y}\|_p^p
    = r\|d\mathbf{A}\vec{y}\|_p^p + \|\mathbf{A}\vec{y}\|_p^p
    \ge rd^p \|\mathbf{A}\vec{y}\|_p^p.
\end{align*}
Hence,
\begin{align*}
    \|\mathbf{A}_{\mathrm{ldl}}\vec{y}\|_p
    \ge r^{1/p} d \|\mathbf{A}\vec{y}\|_p
    \ge td \|\mathbf{A}\vec{y}\|_p,
\end{align*}
because $r \ge t^p$.
Since \Cref{item:min-dist} of \Cref{def:ldl} holds for $(\mathbf{A}, \vec{s}, \ell, \mathbf{T})$, we have $\|\mathbf{A}\vec{y}\|_p \ge \lambda_1^{(p)}(\mathcal{L}(\mathbf{A})) \ge \ell$ for every nonzero $\vec{y} \in \mathbb{Z}^n$.
Therefore, $\|\mathbf{A}_{\mathrm{ldl}}\vec{y}\|_p \ge td\ell = \ell'$.
Thus every nonzero vector in $\mathcal{L}(\mathbf{A}')$ has $\ell_p$-norm at least $\ell'$, and hence $\lambda_1^{(p)}(\mathcal{L}(\mathbf{A}')) \ge \ell'$.
So \Cref{item:ldlp-min-dist} of \Cref{def:ldlp} holds.

Next, we verify \Cref{item:ldlp-conv-hypercube} of \Cref{def:ldlp}.
Assume that \Cref{item:conv-hypercube} of \Cref{def:ldl} holds for
$(\mathbf{A}, \vec{s}, \ell, \mathbf{T})$.
Let $\vec{b} \in \{0,1\}^k$ be arbitrary.
Then there exists $\vec{v} \in (\vec{s} + \mathcal{L}(\mathbf{A})) \cap \mathcal{B}_p^m(\sigma \ell)$ such that $\mathbf{T}\vec{v} = \vec{b}$.
Define
\begin{align*}
\vec{w} := (d\vec{v}, \ldots, d\vec{v}, \vec{v})^T \in \mathbb{Z}^{(r+1)m}.
\end{align*}
We first show that $\vec{w} \in \vec{s}_{\mathrm{ldl}} + \mathcal{L}(\mathbf{A}_{\mathrm{ldl}})$.
Since $\vec{v} \in \vec{s} + \mathcal{L}(\mathbf{A})$, there exists
$\vec{u} \in \mathbb{Z}^n$ such that $\vec{v} = \vec{s} + \mathbf{A}\vec{u}$.
Hence,
\begin{align*}
    \vec{w}
    = (d\vec{s} + d\mathbf{A}\vec{u}, \ldots, d\vec{s} + d\mathbf{A}\vec{u}, \vec{s} + \mathbf{A}\vec{u})^T
    = \vec{s}_{\mathrm{ldl}} + \mathbf{A}_{\mathrm{ldl}}\vec{u}.
\end{align*}
We next bound the norm of $\vec{w}$.
Since $\|\vec{v}\|_p \le \sigma \ell$, we have
\begin{align*}
    \|\vec{w}\|_p^p
    &= r\|d\vec{v}\|_p^p + \|\vec{v}\|_p^p \\
    &\le (rd^p + 1)(\sigma \ell)^p \\
    &\le d^p(r+1)(\sigma \ell)^p \\
    &= d^p(r+1)\alpha^p \frac{t^p}{r+1}\ell^p \\
    &= \alpha^p t^p d^p \ell^p \\
    &= (\alpha \ell')^p,
\end{align*}
where the third inequality uses $d \ge 1$.
Thus, $\vec{w} \in (\vec{s}_{\mathrm{ldl}} + \mathcal{L}(\mathbf{A}_{\mathrm{ldl}})) \cap \mathcal{B}_p^{(r+1)m}(\alpha \ell')$.
Moreover, by the definition of $\mathbf{T}_{\mathrm{ldl}}$, we have $\mathbf{T}_{\mathrm{ldl}}\vec{w} = \mathbf{T}\vec{v} = \vec{b}$.
Finally, let $\vec{u}' := (\vec{0}_{m'}, \vec{w})^T \in \mathbb{Z}^{m' + (r+1)m}$.
Then $\vec{u}' \in (\vec{s}' + \mathcal{L}(\mathbf{A}')) \cap \mathcal{B}_p^{m' + (r+1)m}(\alpha \ell')$ and $\mathbf{T}'\vec{u}' = \vec{b}$.
Since $\vec{b}$ was arbitrary, \Cref{item:ldlp-conv-hypercube} of \Cref{def:ldlp} holds.

Now assume that $(\mathbf{B}, d)$ is a NO instance of coGapSVP$_p$, that is, $\lambda_1^{(p)}(\mathcal{L}(\mathbf{B})) < d$.
Then there exists a nonzero vector $\vec{x} \in \mathbb{Z}^{n'}$ such that $\|\mathbf{B}\vec{x}\|_p < d$.
Let $\vec{z} := (\vec{x}, \vec{0}_n)^T \in \mathbb{Z}^{n'+n}\setminus\{\vec{0}\}$.
Then
\begin{align*}
    \|\mathbf{A}'\vec{z}\|_p
    = \|t\ell \cdot \mathbf{B}\vec{x}\|_p
    = t\ell \cdot \|\mathbf{B}\vec{x}\|_p
    < t\ell d
    = \ell'.
\end{align*}
Therefore, $\lambda_1^{(p)}(\mathcal{L}(\mathbf{A}')) < \ell'$,
so \Cref{item:ldlp-min-dist} of \Cref{def:ldlp} fails.
Hence, $(\mathbf{A}', \vec{s}', \ell', \mathbf{T}') \notin (p,\alpha)\text{-}\LDLP$.
\end{proof}

\subsection{For the Infinity Norm}
Even for $p=\infty$, by going through $\CoeffLDLP$ as an intermediate problem, we show that $(\infty,\alpha)$-$\LDLP$ is complete for the second level of the polynomial hierarchy.

The following theorem is obtained via a reduction from
$\forall\exists$ 1-in-3-SAT to $(\infty,\alpha,\beta)$-$\CoeffLDLP$.

\begin{theorem}
\label{thm:variable-ldlp-is-pi2-hard-infty}
For any constants $\alpha\in(1/2,1)$ and $\beta\in(0,1-\alpha)$, $(\infty,\alpha,\beta)$-$\CoeffLDLP$ is $\Pi_2^p$-hard under deterministic polynomial-time many-one reductions.
\end{theorem}

\begin{proof}
We reduce from $\forall\exists$~1-in-3-SAT, which is $\Pi_2^p$-complete by \Cref{thm:sat}.

Let
\[
    \Phi=\forall \vec{u}\in\{0,1\}^k\ \exists \vec{v}\in\{0,1\}^t:\ \psi(\vec{u},\vec{v})
\]
be an instance of $\forall\exists$~1-in-3-SAT, where $\psi$ is a Boolean formula in which every clause contains exactly three literals. Let $m_\psi$ be the number of clauses of $\psi$, and set $n:=k+t$.

We first describe the output of the reduction.
Apply \Cref{lem:1in3-clause-gadget} to $\psi$. We obtain
\[
    \mathbf{C}\in\mathbb Z^{m_\psi\times n}
    \qquad\text{and}\qquad
    \vec b\in\mathbb Z^{m_\psi}
\]
such that for every $\vec z\in\{0,1\}^n$,
\begin{align}
\label{eq:var-clause-sat}
    \mathbf{C}\vec z=\vec b
    &\iff
    \vec z \text{ satisfies } \psi \text{ in the 1-in-3 sense},\\
\label{eq:var-clause-gap}
    \mathbf{C}\vec z\neq \vec b
    &\Longrightarrow
    \norm{\mathbf{C}\vec z-\vec b}_\infty\ge 2.
\end{align}
Set $m:=n+m_\psi$, $\ell:=2$, and
\[
    \mathbf T:=
    \begin{pmatrix}
        \mathbf I_k & \mathbf O_{k\times t}
    \end{pmatrix}
    \in\mathbb Z^{k\times n}.
\]
We output the tuple $(\mathbf A,\vec s,\ell,\mathbf T)$ defined by
\[
    \mathbf A:=
    \begin{pmatrix}
        2\mathbf I_n\\
        \mathbf C
    \end{pmatrix}
    \in\mathbb Z^{m\times n},
    \qquad
    \vec s:=
    \begin{pmatrix}
        \vec 1_n\\
        \vec b
    \end{pmatrix}
    \in\mathbb Z^m.
\]
This map is computable deterministically in polynomial time.

We verify the promise condition $\lambda_1^{(\infty)}(\mathcal L(\mathbf A))\ge \ell$. The columns of $\mathbf A$ are linearly independent because the first $n$ rows are $2\mathbf I_n$. Let $\vec z\in\mathbb Z^n\setminus\{\vec 0\}$. Then there exists some index $i\in[n]$ such that $z_i\neq 0$, and hence the $i$-th coordinate of $\mathbf A\vec z$ equals $2z_i$. Therefore,
\[
    \norm{\mathbf A\vec z}_\infty\ge |2z_i|\ge 2=\ell.
\]
Thus $\lambda_1^{(\infty)}(\mathcal L(\mathbf A))\ge \ell$.

We next prove completeness. Assume that $\Phi$ is true. Fix an arbitrary vector $\vec y\in\{0,1\}^k$. Since $\Phi$ is true, there exists $\vec v\in\{0,1\}^t$ such that the assignment $\vec z:=(\vec y,\vec v)\in\{0,1\}^n$ satisfies $\psi$ in the 1-in-3 sense. By \Cref{eq:var-clause-sat}, we have $\mathbf C\vec z=\vec b$.
For each $i\in[n]$, the $i$-th coordinate of $2\mathbf I_n\vec z-\vec 1_n$ equals $2z_i-1\in\{-1,1\}$, so its absolute value is $1$. Since the lower block satisfies $\mathbf C\vec z-\vec b=\vec 0$, we obtain
\[
    \norm{\mathbf A\vec z-\vec s}_\infty\le 1.
\]
Because $\alpha>1/2$ and $\ell=2$, we have $\alpha\ell>1$, and therefore $\norm{\mathbf A\vec z-\vec s}_\infty\le \alpha\ell$. Moreover,
\[
    \mathbf T\vec z=\vec y.
\]
Since $\vec y\in\{0,1\}^k$ was arbitrary, we conclude that
\[
    \{0,1\}^k
    \subseteq
    \{\mathbf T\vec z:\vec z\in\mathbb Z^n,\ \norm{\mathbf A\vec z-\vec s}_\infty\le \alpha\ell\}.
\]
Hence $(\mathbf A,\vec s,\ell,\mathbf T)$ is a YES instance of $(\infty,\alpha,\beta)$-$\CoeffLDLP$.

Finally, we prove soundness. Assume that $\Phi$ is false. Then there exists $\vec y^\star\in\{0,1\}^k$ such that for every $\vec v\in\{0,1\}^t$, the assignment $(\vec y^\star,\vec v)$ does not satisfy $\psi$ in the 1-in-3 sense.
We claim that
\[
    \vec y^\star\notin
    \{\mathbf T\vec z:\vec z\in\mathbb Z^n,\ \norm{\mathbf A\vec z-\vec s}_\infty\le (\alpha+\beta)\ell\}.
\]
Suppose toward a contradiction that there exists $\vec z\in\mathbb Z^n$ such that $\mathbf T\vec z=\vec y^\star$ and $\norm{\mathbf A\vec z-\vec s}_\infty\le (\alpha+\beta)\ell$. Since $\alpha+\beta<1$ and $\ell=2$, we have $(\alpha+\beta)\ell<2$.
For each $i\in[n]$, the $i$-th coordinate of $\mathbf A\vec z-\vec s$ equals $2z_i-1$. Hence
\[
    |2z_i-1|\le (\alpha+\beta)\ell<2.
\]
Because $2z_i-1$ is an odd integer, this forces $2z_i-1\in\{-1,1\}$, and therefore $z_i\in\{0,1\}$. Thus $\vec z\in\{0,1\}^n$.
Since $\mathbf T\vec z=\vec y^\star$, the first $k$ coordinates of $\vec z$ are exactly $\vec y^\star$. Therefore $\vec z$ does not satisfy $\psi$ in the 1-in-3 sense. By \Cref{eq:var-clause-sat}, we have $\mathbf C\vec z\neq \vec b$. Hence \Cref{eq:var-clause-gap} implies that
\[
    \norm{\mathbf C\vec z-\vec b}_\infty\ge 2.
\]
Since $\mathbf C\vec z-\vec b$ is the lower block of $\mathbf A\vec z-\vec s$, we obtain
\[
    \norm{\mathbf A\vec z-\vec s}_\infty\ge 2,
\]
which contradicts $\norm{\mathbf A\vec z-\vec s}_\infty\le (\alpha+\beta)\ell<2$.
Thus no such $\vec z$ exists, and $(\mathbf A,\vec s,\ell,\mathbf T)$ is a NO instance of $(\infty,\alpha,\beta)$-$\CoeffLDLP$.
\end{proof}

For a target parameter \(\alpha'\in(1/2,1)\), choose $\alpha := (1/2+\alpha')/2$ and $\beta := \alpha'-\alpha$. Then, together with \Cref{prop:pi2p}, Item 2 of \Cref{thm:ldlp-is-pi2-complete-lp} follows from \Cref{prop:coefficient-to-coset} and \Cref{thm:variable-ldlp-is-pi2-hard-infty}.

\printbibliography

\appendix
\section{Proof of Theorem~\ref{lem:binary-rs-gadget-lp}}
\label{appendix:ldl}
For each prime $q$, fix an ordering $a_1,\ldots,a_q$ of all elements of
$\mathbb{F}_q$. For every integer $k$ with $1\leq k<q$, define
\[
    \mathbf{H}_q(k)
    =
    \begin{pmatrix}
        1 & \cdots & 1\\
        a_1 & \cdots & a_q\\
        \vdots & & \vdots\\
        a_1^{k-1} & \cdots & a_q^{k-1}
    \end{pmatrix}
    \in \mathbb{F}_q^{k\times q}
\]
and
\[
    \mathcal{L}_{q,k}
    :=
    \{\vec{z}\in\mathbb{Z}^q:
      \mathbf{H}_q(k)\vec{z}=\vec{0}\text{ in }\mathbb{F}_q^k\}.
\]
Here and below, when a matrix over $\mathbb{F}_q$ is multiplied by an
integer vector, the entries of the vector are reduced modulo $q$
before the multiplication. Thus $\mathcal{L}_{q,k}$ is defined
entirely as a subset of $\mathbb{Z}^q$ by the displayed congruences.

We first record an elementary property of these Reed--Solomon lattice
cosets that will be used twice below.

\begin{lemma}\label{lem:short-rs-coset-is-binary}
Let $q$ be a prime and let $1\leq k<q$. Let
$\vec{y}\in\{0,1\}^q$ have Hamming weight $h$, meaning that
exactly $h$ coordinates of $\vec{y}$ are equal to $1$, and suppose
that $q>2h$.
If $p\geq\log_2 3$ and
\[
    \vec{z}\in \vec{y}+\mathcal{L}_{q,k}
    \qquad\text{and}\qquad
    \|\vec{z}\|_p^p<h+1,
\]
then $\vec{z}\in\{0,1\}^q$ and $\|\vec{z}\|_1=h$.
\end{lemma}

\begin{proof}
Because the first row of $\mathbf{H}_q(k)$ is the all-one row and
$\vec{z}-\vec{y}\in\mathcal{L}_{q,k}$, we have
\[
    \sum_{i=1}^q z_i\equiv \sum_{i=1}^q y_i=h\pmod q.
\]
Since the coordinates of $\vec{z}$ are integers and $p\geq1$,
\[
    \left|\sum_{i=1}^q z_i\right|
    \leq \sum_{i=1}^q |z_i|
    \leq \sum_{i=1}^q |z_i|^p
    =\|\vec{z}\|_p^p
    <h+1.
\]
Thus $-h\leq\sum_i z_i\leq h$. As $q>2h$, the only integer in this
interval that is congruent to $h$ modulo $q$ is $h$ itself. Hence
\begin{equation}\label{eq:rs-coordinate-sum}
    \sum_{i=1}^q z_i=h.
\end{equation}

For every integer $b$, one has $|b|^p-b\geq0$. Moreover, if
$b\notin\{0,1\}$, then $|b|^p-b\geq1$. Indeed, if $b\leq-1$, then
\[
    |b|^p-b=|b|^p+|b|\geq2.
\]
If $b\geq2$, then
\[
    b^p-b
    =b(b^{p-1}-1)
    \geq 2(2^{p-1}-1)
    =2^p-2
    \geq1,
\]
where the last inequality uses $p\geq\log_2 3$.

If some coordinate of $\vec{z}$ were outside $\{0,1\}$, then
\Cref{eq:rs-coordinate-sum} would give
\[
    \|\vec{z}\|_p^p
    =
    \sum_{i=1}^q z_i
    +
    \sum_{i=1}^q\bigl(|z_i|^p-z_i\bigr)
    \geq h+1,
\]
contrary to the assumption. Therefore $\vec{z}\in\{0,1\}^q$, and
\Cref{eq:rs-coordinate-sum} then implies $\|\vec{z}\|_1=h$.
\end{proof}

\begin{proposition}[{\cite[Proposition~3.2]{wan2026}}]
\label{prop:wan-rs-basis}
Let $q$ be a prime and let $1\leq k<q$. Partition
$\mathbf{H}_q(k)=(\mathbf{V}\ \mathbf{W})$, where
$\mathbf{V}\in\mathbb{F}_q^{k\times k}$ consists of the first $k$
columns and $\mathbf{W}\in\mathbb{F}_q^{k\times(q-k)}$ consists of
the remaining columns. Then $\mathbf{V}$ is invertible. Let
$\widehat{\boldsymbol{\Lambda}}\in\mathbb{Z}^{k\times(q-k)}$ be any
integer matrix whose reduction modulo $q$ is
$\mathbf{V}^{-1}\mathbf{W}$, and define
\[
    \mathbf{B}_{q,k}
    :=
    \begin{pmatrix}
        q\mathbf{I}_k
        &
        -\widehat{\boldsymbol{\Lambda}}
        \\
        \mathbf{0}_{(q-k)\times k}
        &
        \mathbf{I}_{q-k}
    \end{pmatrix}
    \in\mathbb{Z}^{q\times q}.
\]
Then $\mathbf{B}_{q,k}$ is a basis of $\mathcal{L}_{q,k}$; that is,
\[
    \mathcal{L}_{q,k}=\mathbf{B}_{q,k}\mathbb{Z}^q.
\]
\end{proposition}

\begin{proof}
The matrix $\mathbf{V}$ is a Vandermonde matrix whose evaluation
points $a_1,\ldots,a_k$ are distinct, and hence it is invertible. For
every $\vec{u}\in\mathbb{Z}^k$ and
$\vec{v}\in\mathbb{Z}^{q-k}$,
\[
    \mathbf{H}_q(k)\mathbf{B}_{q,k}
    \begin{pmatrix}
        \vec{u}\\
        \vec{v}
    \end{pmatrix}
    \equiv
    -\mathbf{V}
     (\mathbf{V}^{-1}\mathbf{W})\vec{v}
    +
    \mathbf{W}\vec{v}
    =
    \vec{0}
    \pmod q.
\]
Thus $\mathbf{B}_{q,k}\mathbb{Z}^q\subseteq\mathcal{L}_{q,k}$.
Conversely, suppose that
\[
    \begin{pmatrix}
        \vec{u}\\
        \vec{v}
    \end{pmatrix}
    \in\mathcal{L}_{q,k}.
\]
Then
$\mathbf{V}\vec{u}+\mathbf{W}\vec{v}\equiv\vec{0}\pmod q$, and
hence
$\vec{u}+\mathbf{V}^{-1}\mathbf{W}\vec{v}\equiv\vec{0}\pmod q$.
It follows that
$\vec{u}=q\vec{w}-\widehat{\boldsymbol{\Lambda}}\vec{v}$ for some
$\vec{w}\in\mathbb{Z}^k$. Therefore
\[
    \begin{pmatrix}
        \vec{u}\\
        \vec{v}
    \end{pmatrix}
    =
    \mathbf{B}_{q,k}
    \begin{pmatrix}
        \vec{w}\\
        \vec{v}
    \end{pmatrix}.
\]
Thus $\mathcal{L}_{q,k}=\mathbf{B}_{q,k}\mathbb{Z}^q$.
\end{proof}

\begin{lemma}[{\cite[Lemma~3.4]{wan2026}}]
\label{lem:wan-rs-minimum-distance}
Let $q$ be a prime and let $1\leq k\leq q/2$. Then, for every
$1\leq p<\infty$,
\[
    \lambda_1^{(p)}(\mathcal{L}_{q,k})\geq(2k)^{1/p}.
\]
\end{lemma}

\begin{proposition}[{\cite[Theorem~7.3]{wan2026}}]
\label{prop:wan-rs-projection}
Let $0<\varepsilon<1$, and choose constants $\delta$ and
$\varepsilon_1$ satisfying
\[
    0<\delta<\varepsilon_1
    <\frac{\varepsilon}{2(1+\varepsilon)}.
\]
For every sufficiently large prime $q$, define
\[
    k
    :=
    \left\lfloor\frac{q^{\varepsilon_1}}{2}\right\rfloor,
    \qquad
    h
    :=
    \lfloor(1+\varepsilon)k\rfloor,
    \qquad
    R
    :=
    \lfloor q^\delta\rfloor,
\]
and let
\[
    \vec{y}
    :=
    (\underbrace{1,\ldots,1}_{h},0,\ldots,0)^{\mathsf T}
    \in\{0,1\}^q.
\]
Then, for every $1\leq p<\infty$,
\[
    \{0,1\}^R
    \subseteq
    \left\{
        \mathbf{P}_R\vec{z}:
        \vec{z}\in
        (\vec{y}+\mathcal{L}_{q,k})
        \cap
        \mathcal{B}_p^q\left(
            \left(\frac{1+\varepsilon}{2}\right)^{1/p}
            (2k)^{1/p}
        \right)
    \right\},
\]
where
$\mathbf{P}_R=(\mathbf{I}_R\ \mathbf{0}_{R\times(q-R)})$
is the projection onto the first $R$ coordinates.
\end{proposition}

\begin{proof}[Proof of Theorem \ref{lem:binary-rs-gadget-lp}]
Choose a rational constant $\varepsilon$ such that $0<\varepsilon<2\alpha^p-1$, and set
\[
    \beta := \left(\frac{1+\varepsilon}{2}\right)^{1/p}.
\]
Then $0<\varepsilon<1$ and $2^{-1/p}<\beta<\alpha$.
Choose rational constants $\delta$ and $\varepsilon_1$ satisfying
\[
    0<\delta<\varepsilon_1
    <\frac{\varepsilon}{2(1+\varepsilon)}.
\]
For a sufficiently large prime $q$, define
\[
    k
    :=
    \left\lfloor\frac{q^{\varepsilon_1}}{2}\right\rfloor,
    \qquad
    h
    :=
    \lfloor(1+\varepsilon)k\rfloor,
    \qquad
    R
    :=
    \lfloor q^\delta\rfloor.
\]
Let $\mathcal{L}_{q,k}$ be the Reed--Solomon lattice defined at the
beginning of this appendix, and let
\[
    \vec{y}
    :=
    (\underbrace{1,\ldots,1}_{h},0,\ldots,0)^{\mathsf T}
    \in\{0,1\}^q.
\]

We first show that every prescribed projection is realized by a binary
vector of Hamming weight $h$. Put
\[
    p_0:=\log_2 3
    \qquad\text{and}\qquad
    \alpha_0
    :=
    \left(\frac{1+\varepsilon}{2}\right)^{1/p_0}.
\]
Applying \Cref{prop:wan-rs-projection} with $p=p_0$, for all
sufficiently large primes $q$,
\begin{equation}\label{eq:wan-projection}
    \{0,1\}^R
    \subseteq
    \left\{
        \mathbf{P}_R\vec{z}:
        \vec{z}\in
        (\vec{y}+\mathcal{L}_{q,k})
        \cap
        \mathcal{B}_{p_0}^q\bigl(\alpha_0(2k)^{1/p_0}\bigr)
    \right\}.
\end{equation}

Every vector $\vec{z}$ in the intersection on the right satisfies
\[
    \|\vec{z}\|_{p_0}^{p_0}
    \leq
    \alpha_0^{p_0}(2k)
    =
    (1+\varepsilon)k
    <
    h+1.
\]
Moreover, $h<2k\leq q^{\varepsilon_1}$, and hence $q>2h$ for all sufficiently large $q$.
Lemma~\ref{lem:short-rs-coset-is-binary} therefore shows that every
such $\vec{z}$ is binary and has Hamming weight $h$.

Consequently, \Cref{eq:wan-projection} implies
\begin{equation}\label{eq:binary-projection-surjective}
    \{0,1\}^R
    \subseteq
    \left\{
        \mathbf{P}_R\vec{z}:
        \vec{z}\in
        (\vec{y}+\mathcal{L}_{q,k})\cap\{0,1\}^q,
        \ \|\vec{z}\|_1=h
    \right\}.
\end{equation}

We now describe the algorithm on input $r$. Fix a sufficiently large
constant $q_0$, depending only on
$p,\alpha,\varepsilon,\varepsilon_1$, and $\delta$, such that the
preceding conclusions and all inequalities used below hold for every
prime $q\geq q_0$. Let
\[
    Q
    :=
    \max\left\{
        q_0,
        \left\lceil(r+1)^{1/\delta}\right\rceil
    \right\}.
\]
By Bertrand's postulate, there is a prime $q\in[Q,2Q]$.
Such a prime can be found deterministically by testing the integers in
this interval with a deterministic polynomial-time primality test.
In particular,
\[
    q=r^{O(1)}
    \qquad\text{and}\qquad
    R=\lfloor q^\delta\rfloor\geq r.
\]

In \Cref{prop:wan-rs-basis}, choose
$\widehat{\boldsymbol{\Lambda}}$ so that every entry belongs to
$\{0,\ldots,q-1\}$, and let
$\mathbf{B}_{q,k}\in\mathbb{Z}^{q\times q}$ be the resulting basis of
$\mathcal{L}_{q,k}$. Its construction uses only arithmetic and matrix
inversion over $\mathbb{F}_q$ and is deterministic polynomial time.

Set
\[
    a:=\lceil p\rceil,
    \qquad
    m:=(2k)^a,
    \qquad
    N:=mq,
    \qquad
    d:=q.
\]
Define
\[
    \mathbf{A}
    :=
    \begin{pmatrix}
        \mathbf{B}_{q,k}\\
        \vdots\\
        \mathbf{B}_{q,k}
    \end{pmatrix}
    \in\mathbb{Z}^{N\times d},
    \qquad
    \vec{s}
    :=
    \begin{pmatrix}
        \vec{y}\\
        \vdots\\
        \vec{y}
    \end{pmatrix}
    \in\mathbb{Z}^N,
\]
where both displays contain $m$ blocks.
Since $\mathbf{B}_{q,k}$ is nonsingular, $\mathbf{A}$ has full column
rank and is therefore a lattice basis.

Finally, define
\[
    \ell
    :=
    \left\lceil
        \frac{(mh)^{1/p}}{\alpha}
    \right\rceil,
\]
and let
$\mathbf{P}_r=(\mathbf{I}_r\ \mathbf{0}_{r\times(N-r)})
\in\mathbb{Z}^{r\times N}$ be the projection onto the first $r$
coordinates.

Every vector in $\mathcal{L}(\mathbf{A})$ has the form $(\vec{v}^{\mathsf T},\ldots,\vec{v}^{\mathsf T})^{\mathsf T}$
for some $\vec{v}\in\mathcal{L}_{q,k}$.
For all sufficiently large $q$, we have $k\leq q/2$. Therefore,
using \Cref{lem:wan-rs-minimum-distance},
\begin{equation}\label{eq:repeated-minimum-distance}
    \lambda_1^{(p)}(\mathcal{L}(\mathbf{A}))
    =
    m^{1/p}\lambda_1^{(p)}(\mathcal{L}_{q,k})
    \geq
    (2km)^{1/p}.
\end{equation}
Since $h \leq (1+\varepsilon)k = 2\beta^p k$, we have
\[
    \ell
    <
    \frac{(mh)^{1/p}}{\alpha}+1
    \leq
    \frac{\beta}{\alpha}(2km)^{1/p}+1.
\]
Because $\beta<\alpha$ and $(2km)^{1/p}$ tends to infinity with $q$,
we may include in the choice of $q_0$ the condition
\[
    1
    \leq
    \left(1-\frac{\beta}{\alpha}\right)(2km)^{1/p}.
\]
It follows that $\ell\leq(2km)^{1/p}$.
Together with \Cref{eq:repeated-minimum-distance}, this proves
\Cref{item:ldl-1}.

To prove \Cref{item:ldl-2}, fix $\vec{u}\in\{0,1\}^r$. Apply
\Cref{eq:binary-projection-surjective} to the vector in
$\{0,1\}^R$ obtained by appending $R-r$ zeros to $\vec{u}$. There is
$\vec{z}\in(\vec{y}+\mathcal{L}_{q,k})\cap\{0,1\}^q$ with
$\|\vec{z}\|_1=h$ whose first $r$ coordinates are $\vec{u}$.
Define $\vec{x} := (\vec{z}^{\mathsf T},\ldots,\vec{z}^{\mathsf T})^{\mathsf T} \in\mathbb{Z}^N$, with $m$ blocks.
Then
\[
    \vec{x}\in\vec{s}+\mathcal{L}(\mathbf{A}),
    \qquad
    \mathbf{P}_r\vec{x}=\vec{u},
\]
and, because $\vec{z}$ is binary of Hamming weight $h$,
\[
    \|\vec{x}\|_p
    =
    \bigl(m\|\vec{z}\|_p^p\bigr)^{1/p}
    =
    (mh)^{1/p}
    \leq
    \alpha\ell.
\]
This proves \Cref{item:ldl-2}.

It remains to prove the additional assertion. Suppose that $p\geq\log_2 3$.
Since $h<2k$, we have $h+1\leq2k$. By the mean value theorem,
\begin{align*}
    (m(h+1))^{1/p}-(mh)^{1/p}
    &=
    m^{1/p}
    \bigl((h+1)^{1/p}-h^{1/p}\bigr)
    \\
    &\geq
    \frac{m^{1/p}}
         {p(h+1)^{1-1/p}}
    \\
    &\geq
    \frac{(2k)^{a/p}}
         {p(2k)^{1-1/p}}
    \\
    &=
    \frac{(2k)^{(a+1)/p-1}}{p}.
\end{align*}
The exponent
\[
    \frac{a+1}{p}-1
    =
    \frac{a+1-p}{p}
\]
is positive because $a=\lceil p\rceil$. Thus, by increasing $q_0$ if
necessary, the last expression is larger than $\alpha$. Since
$\lceil t\rceil<t+1$ for every real $t$, it follows that
\begin{equation}\label{eq:rounding-below-next-weight}
    \alpha\ell
    <
    (mh)^{1/p}+\alpha
    <
    (m(h+1))^{1/p}.
\end{equation}

Take any $\vec{x}\in(\vec{s}+\mathcal{L}(\mathbf{A})) \cap \mathcal{B}_p^N(\alpha\ell)$.
There is a vector $\vec{z}\in\vec{y}+\mathcal{L}_{q,k}$
such that $\vec{x} = (\vec{z}^{\mathsf T},\ldots,\vec{z}^{\mathsf T})^{\mathsf T}$.
By \Cref{eq:rounding-below-next-weight},
\[
    m\|\vec{z}\|_p^p
    =
    \|\vec{x}\|_p^p
    \leq
    (\alpha\ell)^p
    <
    m(h+1).
\]
Hence $\|\vec{z}\|_p^p<h+1$.
Since $q>2h$, \Cref{lem:short-rs-coset-is-binary} yields $\vec{z}\in\{0,1\}^q$.
Therefore $\vec{x}\in\{0,1\}^N$,
which proves the additional assertion.

Finally, $q=r^{O(1)}$, and
\[
    k,h,R,m,N,d=r^{O(1)}.
\]
Under the standard convention that the fixed norm parameters are
efficiently computable, the integer $\ell$ can be computed exactly in
time polynomial in $\log q$. In particular, this is immediate when
$p$ and $\alpha$ are rational constants.

The entries of $\mathbf{B}_{q,k}$ have $O(\log q)$ bits,
$\mathbf{A}$ consists of $m$ copies of this matrix, and
$\vec{s}$ and $\mathbf{P}_r$ are binary. Moreover,
\[
    \ell
    \leq
    (2km)^{1/p}
    =
    r^{O(1)}.
\]
Hence the total bit length of $(\mathbf{A},\vec{s},\ell,\mathbf{P}_r)$
is bounded by $r^{O(1)}$, and all these objects are produced in
deterministic polynomial time.
\end{proof}
\section{Parametric Honest Hardness of LDLP Implies Locally Dense Lattice Constructions}\label{appendix:meta}
In this appendix, we show a simple meta-observation: if LDLP is hard under reductions that are \emph{honest} with respect to the hypercube dimension, then locally dense lattices can be constructed explicitly.

\begin{definition}[parametric honest reduction to LDLP]
\label{def:k-parametric-honest-reduction}
Let \(L\subseteq\{0,1\}^*\) be a language.
A deterministic polynomial-time many-one reduction \(R\) from \(L\) to \((p,\alpha)\)-\(\LDLP\) is called \emph{parametric honest} if there exists a constant \(\varepsilon>0\) such that for every input \(x\in\{0,1\}^*\), if
\[
    R(x)=(\mathbf A_x,\vec s_x,\ell_x,\mathbf T_x)
    \qquad\text{and}\qquad
    \mathbf T_x\in\mathbb Z^{k_x\times m_x},
\]
then $k_x\ge |x|^\varepsilon$.
In other words, the hypercube dimension of the output LDLP instance is polynomially lower-bounded in the input length.
\end{definition}

The next claim shows that parametric honest hardness of LDLP already implies an explicit construction of locally dense lattices.
The reason is that a YES instance of LDLP with \(k\) rows in its transformation matrix is exactly a \((p,\alpha,k)\)-coset-locally dense lattice.
Hence, by applying a parametric honest many-one reduction to a padded satisfiable SAT instance, one obtains a YES instance of LDLP with large enough hypercube dimension.

\begin{proposition}
\label{thm:parametric-honest-hardness-implies-ldl}
Fix \(p\in[1,\infty]\) and \(\alpha>0\).
Suppose that there exists a parametric honest deterministic polynomial-time many-one reduction from \(\mathrm{SAT}\) to \((p,\alpha)\)-\(\LDLP\).
Then there exists a deterministic algorithm that, given a positive integer \(r\), outputs a \((p,\alpha,r)\)-coset-locally dense lattice.
Moreover, the total bit length of the output is bounded by \(r^{O(1)}\).
\end{proposition}

\begin{proof}
Let $R$ be a parametric honest deterministic polynomial-time many-one reduction from $\SAT$ to $(p, \alpha)$-$\LDLP$.
Let $\varepsilon > 0$ be the constant from \Cref{def:k-parametric-honest-reduction}.

Fix an integer $a > 0$ such that $a\varepsilon > 1$.
By padding a fixed satisfiable Boolean formula with dummy variables or tautological clauses, we may assume without loss of generality that, for every positive integer $r$, we can construct in time $r^{O(1)}$ a satisfiable SAT instance $\phi_r$ such that $|\phi_r| \geq r^a$.

Run the reduction $R$ on $\phi_r$ as input, and write
\[
    R(\phi_r) = (\mathbf{A}, \vec{s}, \ell, \mathbf{T}), \qquad \mathbf{T} \in \mathbb{Z}^{k\times m}.
\]
Since $\phi_r \in \SAT$, $R$ is a many-one reduction to $(p, \alpha)$-LDLP, we have $(\mathbf{A}, \vec{s}, \ell, \mathbf{T}) \in (p, \alpha)$-LDLP.
Hence $\lambda_1^{(p)}(\mathcal{L}(\mathbf{A})) \geq \ell$ and
\[
    \{0,1\}^k \subseteq \left\{\mathbf{T}\vec{v} : \vec{v} \in (\vec{s} + \mathcal{L}(\mathbf{A})) \cap \mathcal{B}_p^m(\alpha \ell)\right\}.
\]
Moreover, by parametric honesty,
\[
    k \geq |\phi_r|^{\varepsilon} \geq r^{a\varepsilon} \geq r.
\]

Let $\mathbf{T'} \in \mathbb{Z}^{r\times m}$ be the matrix consisting of the first $r$ rows of $\mathbf{T}$.
We claim that $(\mathbf{A}, \vec{s}, \ell, \mathbf{T}')$ is a $(p, \alpha, r)$-locally dense lattice.
The minimum distance of $\mathcal{L}(\mathbf{A})$ is unchanged: $\lambda_1^{(p)}
(\mathcal{L}(\mathbf{A})) \geq \ell$.
It remains to verify \Cref{item:conv-hypercube} of \Cref{def:ldl}.
Fix an arbitrary vector $\vec b\in\{0,1\}^r$.
Extend it to a vector in \(\{0,1\}^k\) by setting $\widetilde{\vec b}:=(\vec b,\vec 0_{k-r})$.
Since \((\mathbf A,\vec s,\ell,\mathbf T)\) is a YES instance of
\((p,\alpha)\)-\(\LDLP\), there exists $\vec v\in (\vec s+\mathcal L(\mathbf A)) \cap \mathcal B_p^m(\alpha\ell)$ such that $\mathbf T\vec v=\widetilde{\vec b}$.
Taking the first \(r\) coordinates gives
\[
    \mathbf{T}'\vec v=\vec b.
\]
Since \(\vec b\in\{0,1\}^r\) was arbitrary,
\[
    \{0,1\}^{r}
    \subseteq
    \left\{
        \mathbf {T}'\vec v:
        \vec v\in
        (\vec s+\mathcal L(\mathbf A))
        \cap \mathcal B_p^m(\alpha\ell)
    \right\}.
\]
Thus $(\mathbf A,\vec s,\ell,\mathbf{T}')$ is a \((p,\alpha,r)\)-coset-locally dense lattice.

Finally, the output size is \(r^{O(1)}\).
Indeed, \(|\phi_r|=r^{O(1)}\), and \(R\) runs in polynomial time, so the bit length of \(R(\phi_r)\) is polynomial in \(|\phi_r|\), hence \(r^{O(1)}\).
\end{proof}

\end{document}